\documentclass[11pt]{amsart}
\usepackage[utf8]{inputenc}
\usepackage{amsthm,amsmath,amssymb,bbm,hyperref,graphicx,float,tikz}
\usepackage{xcolor}
\usepackage{braket}
\usepackage[normalem]{ulem}
\usepackage{fnpct}
\allowdisplaybreaks

\numberwithin{equation}{section}

\newtheorem{thm}{Theorem}[section]

\newtheorem{lemma}[thm]{Lemma}

\theoremstyle{definition}
\newtheorem{rmk}[thm]{Remark}
\newtheorem{defn}[thm]{Definition}
\newtheorem{assumption}[thm]{Assumption}

\DeclareMathOperator{\tr}{tr}
\newcommand{\Hilbert}{\mathcal{H}}
\newcommand{\smc}{\mathrm{sc}}

\newcommand{\IM}{\textup{Im\,}}

\newcommand{\dapp}{\mathcal{M}}

\newcommand{\R}{\mathbb{R}}
\newcommand{\C}{\mathbb{C}}
\newcommand{\N}{\mathbb{N}}

\newcommand{\ii}{\mathrm{i}}
\newcommand{\ee}{\mathrm{e}}
\newcommand{\dd}{\mathrm{d}}

\usepackage{comment}

\title{Normal Typicality and Dynamical Typicality for a  Random Block-Band Matrix Model}

	\date{\today}

\begin{document}
	\maketitle

	\vspace{0.25cm}
	
	\renewcommand{\thefootnote}{\fnsymbol{footnote}}

	\noindent
	\mbox{}%
	\hfill%
	\begin{minipage}{0.25\textwidth}
		\centering
		{L\'aszl\'o Erd\H{o}s}\footnotemark[1]\\
		\footnotesize{\textit{lerdos@ist.ac.at}}
	\end{minipage}
	\hfill%
	\begin{minipage}{0.25\textwidth}
		\centering
		{Joscha Henheik}\footnotemark[2]\\
		\footnotesize{\textit{joscha.henheik@unige.ch}}
	\end{minipage}
	\hfill%
	\begin{minipage}{0.25\textwidth}
		\centering
		{Cornelia Vogel}\footnotemark[1]\footnotemark[3]\\
		\footnotesize{\textit{cornelia.vogel@math.lmu.de}}
	\end{minipage}
	\hfill%
	\mbox{}%
	\footnotetext[1]{Institute of Science and Technology Austria, Am Campus 1, 3400 Klosterneuburg, Austria. 
	}
\footnotetext[2]{Department of Mathematics, University of Geneva, Rue de Conseil Général 7-9, 1205 Geneva, Switzerland.}
	\footnotetext[3]{Department of Mathematics, LMU Munich, Theresienstr. 39, 80333 Munich, Germany.}
	
	\renewcommand*{\thefootnote}{\arabic{footnote}}
	\vspace{0.25cm}

	\begin{abstract} 
We prove \emph{normal typicality} and \emph{dynamical typicality} for a (centered) random block-band matrix model with block-dependent variances. A key feature of our model is that we achieve intermediate equilibration times, an aspect that has not been proven rigorously in any model before. Our proof builds on recently established concentration estimates for products of resolvents of Wigner type random matrices \cite{ER24} and an intricate analysis of the deterministic approximation. 
	\end{abstract}
~	\\[1mm]
{\footnotesize	Keywords: Quantum dynamics, equilibration, normal typicality, dynamical typicality, Wigner type matrix \\
	MSC Classes: 60B20, 82C10}
\section{Introduction}
We consider a closed macroscopic quantum system in a pure state $\psi_0$ which evolves unitarily according to $\psi_t=\ee^{-\ii Ht}\psi_0$ where $H$ is the Hamilton operator on a large finite dimensional
Hilbert space $\Hilbert$.  Following von Neumann~\cite{vonNeumann29}, we take a fixed orthogonal decomposition of 
$\Hilbert$ into subspaces $\Hilbert_\nu$ (so called \emph{macro spaces}) corresponding to different \emph{macro states} $\nu$,
\begin{align} \label{eq:msdecomp}
	\Hilbert=\bigoplus_\nu\Hilbert_\nu.
\end{align}
 Usually, one of these subspaces has by far the largest dimension and, by analogy to classical mechanics, corresponds to the \emph{thermal equilibrium macro space}.

Let $P_\nu$ denote the  orthogonal  
projection to $\Hilbert_\nu$. Von Neumann~\cite{vonNeumann29} showed that for random Hamiltonians of the form $H=U^* D U$, where $D$ is a deterministic diagonal matrix satisfying certain non-degeneracy assumptions and $U$ is Haar-distributed, every $\psi_0\in\mathbb{S}(\Hilbert)=\{\phi\in\Hilbert: \|\phi\|=1\}$ evolves so that for most times $t\geq 0$,
\begin{align} \label{eq:normaltyp}
	\|P_\nu\psi_t\|^2 \approx \frac{d_\nu}{N} \quad \forall \nu,
\end{align}
provided that $d_\nu:=\dim\Hilbert_\nu$ and $N:=\dim\Hilbert$ are sufficiently large (and under some further not very restrictive assumptions). This phenomenon is called \emph{normal typicality} (a.k.a.~von Neumann's \emph{quantum ergodic theorem}) and it was rediscovered and strengthened in \cite{GLMTZ10,GLTZ10,GLMTZ10b}. 

However, in von Neumann's setting, the energy eigenbasis of $H$ is unrelated to the orthogonal decomposition of the Hilbert space and as a consequence, for any initial state equilibration takes place almost immediately, see, e.g., \cite{GHT13,GHT14,GHT15}. In contrast to that, in very general settings, bounds on equilibration time can easily exceed the age of the universe \cite{SF12,TTV23a}. 
In this paper, we study an intermediate structured model given
by a random block-band matrix ensemble, with the blocks aligned to the macro space decomposition \eqref{eq:msdecomp}. Hence, in this model, we expect equilibration to take place on different and non-trivial time scales.

More precisely, in the following we model the Hamiltonian $H$ by a random $ N \times N$ matrix with a block structure in a basis that diagonalizes the projections $P_\nu$. We consider two prototypical
models, a $2\times 2$ and a $3\times 3$ block matrix, with blocks of very different sizes {(i.e., corresponding to macro spaces of very different dimensions)}, encoded in a small {$N$-independent} parameter $\lambda \in (0,1)$. The entries of $H$ are independent (up to conjugate symmetry) centered random variables with variances depending\footnote{\label{ftn:WT}In other words, the models we consider are centered \emph{Wigner type random matrices} \cite{AEK17, ER24}, allowing for non-constant variances, and, moreover, for zero-blocks of size $\lfloor cN\rfloor \times \lfloor cN\rfloor$ for $c<1$.} on the block; see Assumptions~\ref{ass:2} and~\ref{ass:3} for the precise definitions.

\medskip 

 We are mainly interested in the behavior of the curves $t\mapsto \|P_\nu\psi_t\|^2$ with initial state $\psi_0\in\mathbb{S}(\Hilbert_\mu)$ for different macro states $\mu, \nu$. They describe the quantum 
 probability that an initial state from the macro space $\mu$ ends up in the macro space $\nu$ after
 time $t$.
  Besides $t\mapsto \|P_\nu\psi_t\|^2$ that is strongly fluctuating \cite{TTV23a, TTV23b}, we also consider its version averaged over the initial state, i.e., we define
 \begin{equation}
w_{\mu \nu}(t) := \mathbb{E}_\mu{\|P_\nu\psi_t\|^2} =\frac{1}{d_\mu} \tr[P_\mu \exp(\ii Ht)P_\nu\exp(-\ii Ht)]\,, 
 \end{equation}
 where the expectation is taken with respect to the uniform distribution on $\mathbb{S}(\Hilbert_\mu)$ for $\psi_0$. In fact, we will prove stronger concentration bounds for  $w_{\mu \nu}(t)$ than for $\|P_\nu\psi_t\|^2$.  Moreover, we will show that the curves $t \mapsto \|P_\nu\psi_t\|^2$ and $t \mapsto w_{\mu \nu}(t)$ exhibit universal behavior, a feature that is called \emph{dynamical typicality} \cite{BG09,MGE11,BRGSR18,Reimann18a,Reimann18b,RG20}.  Additionally, we also study the  infinite time average\footnote{For a general  function $f:[0,\infty)\to\mathbb{C}$, we denote the infinite time average by $\overline{f(t)} = \lim_{T\to\infty} T^{-1}\int_0^T f(t)\, \mathrm{d} t$, whenever this limit exists. }
 \begin{align} \label{eq:dapp}
 	\dapp_{\mu\nu} := \overline{w_{\mu \nu}(t)} = \frac{1}{d_\mu}\sum_e \tr(\Pi_e P_\mu \Pi_e P_\nu) \,,
 \end{align}
where $\Pi_e$ denotes the projection onto the eigenspace of $H$ with eigenvalue $e$. \normalcolor We prove that, for long times $\|P_\nu\psi_t\|^2 \approx  \dapp_{\mu\nu}$, which is called \emph{generalized normal typicality} \cite{TTV23a}.
 
 To summarize, we investigate $t\mapsto \|P_\nu\psi_t\|^2$ on three levels of complexity, schematically indicated as
 	\begin{equation} \label{eq:3obj}
 	\Vert P_\nu \psi_t \Vert^2  \xrightarrow{\psi_0 \in \mathbb{S}(\Hilbert_\mu) \ \text{average}} w_{\mu \nu}(t) \xrightarrow{t \ \text{average}} \mathcal{M}_{\mu \nu}\,. 
 \end{equation}
 We now describe our results for the three objects in \eqref{eq:3obj} in the two random matrix models considered in this paper in more detail.  
 
 Our first result (Theorems~\ref{thm: dappmunu} and~\ref{thm: dappmunu3}, respectively) shows that with very high probability, $\dapp_{\mu \nu} \approx d_\nu/N$, i.e.~our models exhibit \emph{normal typicality} as discussed around \eqref{eq:normaltyp}. 
 
 Our second result (Theorems~\ref{thm: wmunu} and~\ref{thm: wmunu3}, respectively) shows that with very high probability $w_{\mu\nu}(t) \approx d_\nu/N + f_{\mu \nu}(t, \lambda)$ for some function $f_{\mu \nu}(t, \lambda)$, that we identify to leading order, 
 in the regime where  $t$ is sufficiently large and $\lambda$ sufficiently small. In particular, we find that $f_{\mu \nu}(t, \lambda) \to 0$ as $t \to \infty$. From the more precise formulas, we can immediately read off the equilibration time scale, which heavily depends on the dimensions of the two involved macro spaces $\Hilbert_\mu, \Hilbert_\nu$ and thus, in particular, on the small parameter $\lambda$ encoding their relative sizes. Finding tunable (by the parameter $\lambda$) equilibration times is one of the main features of our model. This allows to interpolate between the extremely long/short (and thus rather unphysical) time scales discussed above, the only cases  for which rigorous results existed before.
 
In our third and final result (Theorems~\ref{thm: Pnupsi} and~\ref{thm: Pnupsi3}, respectively), we also study the quantities $\|P_\nu\psi_t\|^2$ for a fixed $\psi_0 \in \mathbb{S}(\mathcal{H}_\mu)$ and find the same asymptotics $\|P_\nu\psi_t\|^2 \approx d_\nu/N + f_{\mu \nu}(\lambda, t)$. All our approximate statements are understood to be valid in the way that we first take $N \to \infty$ (thermodynamic limit) and afterwards $t \to \infty$ and $\lambda \to 0$.

\medskip

The proofs of our main results are based on concentration estimates for products of resolvents, $G(z) := (H-z)^{-1}$ for $z \in \C \setminus \R$, of the random matrices $H$ considered. In fact, it is a remarkable feature of random matrices, that products of resolvents tend to become deterministic as the matrix size goes to infinity. The resolvent products concentrate around a nontrivial leading term and a fluctuation estimate around that determinstic leading term is called a \emph{local law}. The general theory of \emph{multi-resolvent local laws} has systematically been developed in last few years, starting from completely mean-field \emph{Wigner matrices} \cite{ETHpaper, multiG, A2, edgeETH, OTOC} and their deformations (i.e.~with non-zero expectation matrix) \cite{equipart, decor}. 

More precisely, our proofs of Theorems \ref{thm: dappmunu} and \ref{thm: dappmunu3} crucially rest on an important 
consequence  of a recently established two-resolvent local law for \emph{Wigner type matrices} (characterized by a non-constant variance profile, see Footnote~\ref{ftn:WT}) in \cite{ER24}. This consequence is the \emph{Eigenstate Thermalization Hypothesis} (ETH), originally introduced by Deutsch and Srednicki \cite{Deutsch, Srednicki} in the 1990's, which has since then played a fundamental role in the question of thermalization in closed quantum mechanical systems. In our context, the ETH for Wigner type matrices in \cite{ER24} (see \cite{ETHpaper} for the first result on Wigner matrices) allows us to control the overlaps $\tr(\Pi_e P_\mu \Pi_e P_\nu)$ of the (random) eigenprojections $\Pi_e$ with the (deterministic) macro space projectors $P_\mu, P_\nu$ in \eqref{eq:dapp}. 

In order to treat the time dependent quantities $\Vert P_\nu \psi_t\Vert^2$
 and $w_{\mu \nu}(t)$ 
  in \eqref{eq:3obj}, we express the time evolution via a suitable contour integral
\begin{equation*}
\ee^{\ii t H} = \frac{1}{2 \pi \ii} \oint_\gamma \ee^{\ii t z} G(z) \dd z \,, 
\end{equation*}
where $\gamma$ encircles the spectrum of $H$. After application of appropriate two-resolvent local laws, establishing the approximation $d_\nu/N + f_{\mu \nu}(\lambda, t)$ boils down to computing a (double) contour integral of the deterministic approximation to the product of resolvents. Extracting the time-dependent term in the approximation $d_\nu/N + f_{\mu \nu}(\lambda, t)$ in presence of the small parameter $\lambda$ requires a delicate stationary phase approximation with a singular integrand. We point out that the analysis required in our current setting is very different from \cite{pretherm, echo}, where equilibration and thermalization in presence of a small parameter has been investigated: In  \cite{pretherm, echo}, the authors considered matrices of the form $H_0 + \lambda W$, with $W$ being a mean-field Wigner matrix and $H_0$ a deterministic deformation. Hence, in these works, $\lambda$ models the strength of the mean-field noise of the model. In contrast to that, in the present work, $\lambda$ encodes the inhomogeneity of the noise throughout the different blocks. 

\medskip

We conclude this introduction by commenting on previous results on dynamical typicality and normal typicality. For quite general Hamiltonians it was shown in \cite{TTV23a} that for any $\mu$, most $\psi_0\in\mathbb{S}(\Hilbert_\mu)$ are such that for most $t\geq 0$,
\begin{align}
	\|P_\nu\psi_t\|^2 \approx \dapp_{\mu \nu}\quad \forall \nu,
\end{align}
provided that $d_\mu$ is large and the eigenvalues and eigenvalue gaps of $H$ are not too highly degenerate. While \cite{TTV23a} was only concerned with absolute errors, in further development it was shown in \cite{TTV23b} that the relative errors are small as well if $H$ is modeled by a suitable random matrix. 
It was also conjectured that normal typicality, i.e.~$\dapp_{\mu\nu}\approx d_\nu/N$, holds, but only a
 lower bound $\dapp_{\mu\nu}\gtrsim (d_\nu/N)^{16}$ was proved instead.
 As to dynamical typicality, the approximation $\|P_\nu\psi_t\|^2 \approx  w_{\mu \nu}(t) $ was  also
 shown in~\cite{TTV23a} as the dimensions are large, $d_\mu\to\infty$, but no specific formula was given for $w_{\mu \nu}(t) $. In particular no analysis of the equilibration times was available.   We remark that a similar result was 
 rigorously  obtained in~\cite{MGE11,BRGSR18}, see~\cite{TTV23a}[Section~2.2] for more details.
 Summarizing, the current work provides a detailed description of a concrete model featuring an intermediate equilibration time scale,
 while previous papers considered more general Hamiltonians with much less explicit results.

\medskip

The remainder of this paper is organized as follows: In Section~\ref{sec:main} we formulate and discuss our main results first for the $2\times 2$ block model and later for the $3\times 3$ block model. In Section~\ref{sec:2proof} we give the proofs of the results regarding the case of two macro spaces and in Section~\ref{sec:3proof} we prove analogous results for the model with three macro spaces. Some additional proofs are given in Appendix~\ref{app:aux}.

\subsection*{Notations} Let $\Hilbert$ be a Hilbert space of dimension $N=\dim\Hilbert<\infty$, i.e.~$\mathcal{H} = \C^N$. For a vector $\psi\in\Hilbert$ we denote by $\|\psi\|$ the Hilbert space norm of $\psi$. Moreover, we write 
\begin{align*}
	\mathbb{S}(\Hilbert) = \{\phi\in\Hilbert: \|\phi\|=1\}
\end{align*}
for the unit sphere of $\Hilbert$.
For the spectrum of an operator $A$ on $\Hilbert$ we use the notation $\mathrm{spec}(A)$ and its normalized trace is denoted by $\langle A\rangle := N^{-1} \tr(A)$. For positive integers $k,l\in\mathbb{N}$ we write $E_{k,l}$ for the $k\times l$ matrix whose entries are all equal to $1$ and $\mathbf{1}$ is the identity on $\Hilbert$. Furthermore, for $k \in \N$ we denote the set of positive integers up to $k$ by $[k] := \{1,2,..., k\}$. For vectors $\psi,\phi\in\Hilbert=\mathbb{C}^N$ and matrices $A\in\mathbb{C}^{N\times N}$ we define
\begin{align*}
	\langle\psi|\phi\rangle := \sum_i \overline{\psi_i}\phi_i, \qquad \langle\psi|A|\phi\rangle := \langle\psi|A\phi\rangle.
\end{align*}
Moreover, for $\lambda>0$ and $t\in\mathbb{R}$ we denote by $C(\lambda)$ and $C(\lambda,t)$ constants depending only on their arguments and whose precise values are irrelevant and  might change from line to line.

Additionally, we will need the notion of \textit{stochastic domination} (see, e.g., \cite{semicirclegeneral}): 
\begin{defn}[Stochastic domination] Let $X=(X_N)$ and $Y=(Y_N)$ be two sequences of non-negative random variables. We say that $Y$ stochastically dominates $X$ uniformly in $N$ and write $X\prec Y$ if for every $\varepsilon>0$ and $\gamma>0$ there exists $N_0=N_0(\varepsilon,\gamma)\in\mathbb{N}$ such that for all $N\geq N_0(\varepsilon,\gamma)$,
	\begin{align}
		\mathbb{P}\left(X_N\geq N^\varepsilon Y_N\right)< N^{-\gamma}.
	\end{align}
	For complex-valued $X$ satisfying $|X|\prec Y$, we write $X=\mathcal{O}_{\prec}(Y)$.
\end{defn}
Moreover, we say that an event holds ``with very high probability'' if for any fixed $\gamma>0$ the probability of the event is bigger than $1-N^{-\gamma}$ for $N\geq N_0(\gamma)$.

\section{Main Results} \label{sec:main}

In this section, we present our main results for two and three macro spaces in Sections \ref{subsec:main2} and \ref{subsec:main3}, respectively.
\subsection{Two macro spaces} \label{subsec:main2}Let $\lambda>0$ be a small ($N$-independent) parameter, let $D\geq 1$ and let $N = (1+\lambda)D$ be the dimension of the whole Hilbert space $\Hilbert$.\footnote{To avoid technicalities, we shall henceforth suppose that $N$ is in fact an integer.} We partition $\Hilbert$ into two macro spaces of dimension $d_1=\lambda D$ and $d_2 = D$, i.e., 
\begin{equation*}
\Hilbert = \mathbb{C}^{N} \simeq \mathbb{C}^{d_1}\oplus\mathbb{C}^{d_2} =: \Hilbert_1 \oplus \Hilbert_2 \,. 
\end{equation*}
We model the Hamiltonian by a Hermitian random matrix $H=(h_{ij})_{i,j \in [N]}$ whose entries $h_{ij} = \overline{h_{ji}}$ satisfy the following conditions.  
\begin{assumption}[$2 \times 2$ block matrix model] \label{ass:2}
For $i \le j$, the entries $h_{ij}$ are centered and independent random variables. The variance matrix $S=(\mathbb{E}|h_{ij}|^2)_{i,j}$ is given by
\begin{equation*}
	S = \left(\begin{matrix}
		S_{11} & S_{12}\\
		S_{21} & S_{22}
	\end{matrix}\right),
\end{equation*}
where
\begin{equation} \label{eq:2variances}
	S_{11} = \frac{1}{N\lambda} E_{d_1,d_1}, \quad S_{12} = S_{21}^* = \frac{\lambda}{N} E_{d_1,d_2}, \quad S_{22} = \frac{1+\lambda-\lambda^2}{N} E_{d_2,d_2}.
\end{equation}
Moreover, we assume that all centered moments of $\sqrt{N} H$ are uniformly bounded in $N$, i.e., for all $p\in\mathbb{N}$ there exists a ($\lambda$-dependent) constant $C_p(\lambda)>0$ such that, uniformly in $1\leq i,j\leq N$,
\begin{align}
	\mathbb{E}|\sqrt{N} h_{ij}|^p \leq C_p(\lambda). \label{ineq: mom unif bdd}
\end{align} 
\end{assumption}

We remark that the variances in the blocks are chosen such that the rows and columns of $S$ sum up to one ($S$ is a doubly stochastic matrix). As a consequence, the solution of the corresponding quadratic vector equation \cite{AEK19} is given by the Stieltjes transform of the semicircular density (see the proof of Theorem~\ref{thm: dappmunu}). Having this explicit solution facilitates the computations and therefore in the present work we restrict ourselves to this setting. However, we believe that the conclusions we draw concerning the quantities $\dapp_{\mu\nu}$, $w_{\mu\nu}(t)$ and $\|P_\nu\psi_t\|^2$ should be very similar also in a more general setting where the variances in the two diagonal blocks are of order $1/(N\lambda)$ and $1/N$, respectively, while in the off-diagonal blocks they are of order $\lambda/N$ but not fine-tuned to obtain exactly the Stieltjes transform of the semicircular density of states as the solution of the quadratic vector equation. In this case, however, the rigorous proof would be more cumbersome.

We now explain the physical reason for the choice of the $\lambda$-scaling in the three blocks in~\eqref{eq:2variances}. Ignoring the overall $1/N$ factor, notice that the variances in the $S_{11}$ block are bigger by a factor $1/\lambda$ than those in the $S_{22}$ block; this compensates for the different dimensions of the two diagonal blocks and guarantees that typical states in both macro spaces live roughly on the same energy scale. The off-diagonal blocks weakly couple the two macro spaces and allow for transitions between them. The coupling parameter in the off-diagonal blocks determines the time scale of this transition; with our choice it is of order $1/\lambda$. In a more general setup, the size of the off-diagonal elements, hence the transition time scale, could be chosen independently of the $\lambda$ parameter describing the relative sizes of the macro spaces. However, for easier calculation, we have chosen the simplest model where these two parameters are set to be the same.

\subsubsection{Results for two macro spaces}For the $2 \times 2$ block matrix model defined above we have the following results. Their proofs are given in Section \ref{sec:2proof}. 

\begin{thm}[Normal typicality: $\dapp_{\mu \nu} \approx d_\nu/N$]\label{thm: dappmunu} Let $H$ be a random matrix satisfying Assumption \ref{ass:2} and take $\mu, \nu \in \{1,2\}$. 
Denote
	\begin{equation} \label{eq:Mmunudef}
		\dapp_{\mu\nu} = \frac{1}{d_\mu} \sum_{e \in \mathrm{spec}(H)} \tr(\Pi_e P_\mu \Pi_e P_\nu),
	\end{equation}
	where $\Pi_e$ is the projector on the eigenspace of $H$ corresponding to an eigenvalue $e \in \R$, and $P_\mu$ is the projector on $\Hilbert_\mu \subset \Hilbert$. 
Then it holds that
\begin{equation*}
	\dapp_{\mu \nu} = \frac{d_\nu}{N}+\mathcal{O}_\prec(C(\lambda)/\sqrt{N}) \,. 
\end{equation*}
\end{thm}

\begin{thm}[Dynamical typicality: The $w_{\mu \nu}$'s]\label{thm: wmunu}
	Let $H$ be a random matrix satisfying Assumption \ref{ass:2} and take $\mu, \nu \in \{1,2\}$. Denote 
	\begin{equation*}
w_{\mu\nu}(t) = \frac{1}{d_\mu} \tr\left[P_\mu \exp(\ii tH)P_\nu \exp(-\ii t H)\right]
	\end{equation*}
for $t \ge 0$, where $P_\mu$ is the projector on $\Hilbert_\mu \subset \Hilbert$. Then, it holds that 
\begin{equation}  \label{eq:wmunus}
	w_{12}(t) = \frac{d_2}{N} - \frac{1}{\pi (\lambda t)^3}\left(1+o(1)\right) + \mathcal{O}_\prec\left(C(\lambda, t)/N \right),
\end{equation}
where $o(1)$ denotes a quantity vanishing in the limit $t \to \infty$, $\lambda  \to 0$, and $t \lambda \to \infty$, and $C(\lambda, t)$ is a constant depending only on its arguments. 
\end{thm}

The other $w_{\mu \nu}$'s can easily be obtained from $w_{12}$ using trivial symmetries and summation rules (cf.~Lemma \ref{lem: rel wmunu}), namely $\sum_\nu w_{\mu \nu}(t) = 1$ and $w_{\mu \nu}(t) = \tfrac{d_\nu}{d_\mu}w_{\nu \mu}(-t)$ together with the fact that $- H$ satisfies Assumption \ref{ass:2} as well. 

\begin{rmk}[Explicit form of the $w_{\mu \nu}$'s]  \label{rmk:precise}
Our proof of Theorem \ref{thm: wmunu} actually yields explicit expressions of the
$w_{\mu \nu}$'s in \eqref{eq:wmunus} up to an error term of order $1/N$. For example, for $w_{12}$, we have that
\begin{equation} \label{eq:w12precise}
w_{12}(t) = \frac{d_2}{N} \left( 1 - \sum_{n \ge 0} (1-\lambda)^n (n+1)^2 \left(\frac{J_{n+1}(2t)}{t}\right)^2 \right) + \mathcal{O}_\prec \left(C(\lambda, t)/N\right) \,,
\end{equation}
where $J_{n}$ is the $n$-th order Bessel function of the first kind.  The formulas for the other $w_{\mu \nu}$'s are similar. During the proof, we evaluate the sum in \eqref{eq:w12precise} asymptotically to obtain \eqref{eq:wmunus}. 
\end{rmk}

\begin{thm}[Approach to equilibrium]\label{thm: Pnupsi}
		Let $H$ be a random matrix satisfying Assumption~\ref{ass:2} and take $\mu, \nu \in \{1,2\}$.
Let  $\psi_0\in\mathbb{S}(\Hilbert_\mu)$ and denote $\psi_t := \ee^{- \ii t H} \psi_0$ for $t \ge 0$. Then, denoting the projector on the Hilbert space $\Hilbert_\nu \subset \Hilbert$ by $P_\nu$, it holds that 
\begin{subequations} \label{eq:Pnupsis}
\begin{align}
	\Vert P_2 \psi_t \Vert^2 = \frac{d_2}{N} - \frac{1}{\pi(\lambda t)^3} \left(1+o(1)\right) +  \mathcal{O}_\prec\left(C(\lambda, t)/\sqrt{N} \right) \quad \text{for} \quad \psi_0 \in \mathbb{S}(\Hilbert_1),
\end{align}
and 
\begin{align}
	\Vert P_1 \psi_t \Vert^2  = \frac{d_1}{N}- \frac{1}{\pi\lambda^2 t^3}\left(1+o(1)\right) + \mathcal{O}_\prec\left(C(\lambda, t)/\sqrt{N}\right) \quad \text{for} \quad \psi_0 \in \mathbb{S}(\Hilbert_2)\,. 
\end{align}
\end{subequations}
 Here, $o(1)$ denotes a quantity vanishing in the limit $t \to \infty$, $\lambda  \to 0$, and $t \lambda \to \infty$, and $C(\lambda, t)$ is a constant depending only on its arguments. 
\end{thm}
The analogs of the expressions in \eqref{eq:Pnupsis} with $P_1$ instead of $P_2$ and $P_2$ instead of $P_1$, respectively, can easily be obtained via the sum rule $\sum_{\nu} \Vert P_\nu \psi_t \Vert^2 = 1$. {Moreover, during the proof of Theorem \ref{thm: Pnupsi} we obtain similar closed expressions for the $\Vert P_\nu \psi_t\Vert^2$'s as in Remark \ref{rmk:precise}.}

\subsubsection{Discussion} 
Theorem~\ref{thm: dappmunu} shows normal typicality, i.e. 
that $\dapp_{\mu\nu} \approx d_\nu/N$ with very high probability.
Both Theorem~\ref{thm: wmunu} and~\ref{thm: Pnupsi} are statements about dynamical typicality,
the latter holds for any fixed initial data $\psi_0$ with an error of order $O(1/\sqrt{N})$,
 while in the former the initial state
is  averaged out and a more precise $O(1/N)$ error term is obtained.
As the initial state is purely in one subspace, for times smaller than the equilibration time, its part in the other subspace is expected to be smaller than its equilibrium value. The minus signs in 
Theorem~\ref{thm: Pnupsi} show exactly this: the equilibrium value of $\|P_\nu\psi_t\|^2$ with $\psi_0\in\mathbb{S}(\Hilbert_\mu)$, $\mu\neq \nu$, is approached from below, see Figure~\ref{fig:2blocks} for a numerical simulation.

\begin{figure}[h]\label{fig:2blocks}
	\centering
	\includegraphics[height=9cm]{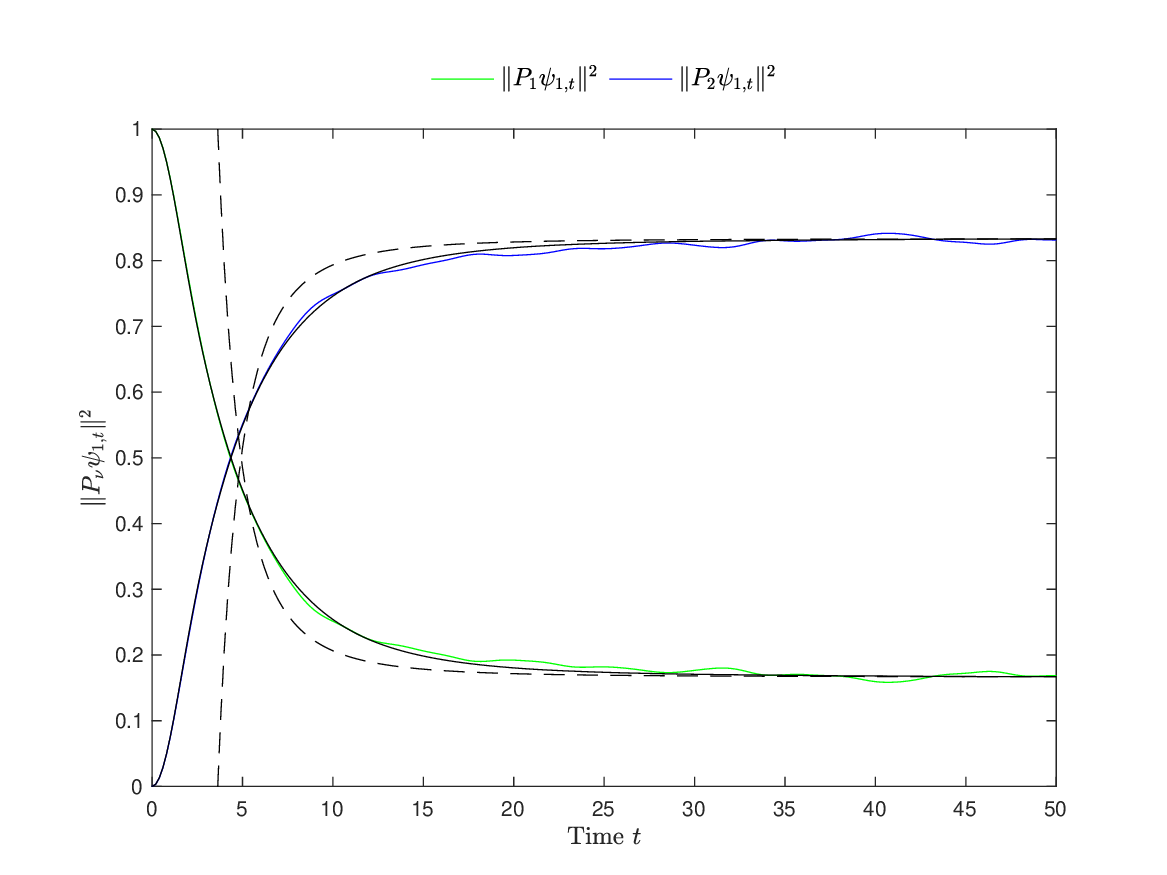}
	\caption{Numerical simulation of the functions $t\mapsto \|P_\nu\psi_t\|^2$ for a random $2\times 2$ block matrix as in Assumption~\ref{ass:2}. Here, $\lambda=0.2$ and the Hilbert space $\Hilbert$ of dimension $N=4200$ is decomposed into 2 macro spaces of dimensions $d_1 = 700$ (green curve) and $d_2=3500$ (blue curve). The initial state $\psi_0\in\mathbb{S}(\Hilbert_1)$ was chosen purely randomly. The black solid curves are the deterministic {(w.r.t.~the randomly chosen $\psi_0$)} approximations $w_{1\nu}(t)$ from \eqref{eq:wmunus} and the black dashed curves are the approximations of $\|P_\nu\psi_t\|^2$ according to Theorem~\ref{thm: Pnupsi}. Note that the dashed curves start from $t\geq 1/(2\lambda) = 2.5$ as the approximations in Theorem~\ref{thm: Pnupsi} are only meaningful for $t \gg 1/\lambda$.}
\end{figure}

Since the equilibration takes place on time scales $t\gg \lambda^{-1}$, we see that by choosing the parameter $\lambda$ in a suitable way, our model also covers the case of equilibration times that are physically more realistic than the ones observed, e.g., in \cite{SF12,GHT13,GHT14,GHT15}.

As mentioned below Theorem~\ref{thm: wmunu}, all other $w_{\mu\nu}$'s can be obtained from $w_{12}$ with the help of Lemma~\ref{lem: rel wmunu}. Another reason why we did not write down all four formulas is that they could be interpreted to carry misleading information. For example, the formula $w_{21}(t)=\lambda w_{12}(-t)$ might suggest that equilibration takes place for times $t\gg \lambda^{-2/3}$. However, relative to the equilibrium value which is of order $\lambda$ for $w_{21}$ (instead of order one for $w_{12}$), we again need that $t\gg \lambda^{-1}$ to ensure that the {time dependent error} term is small compared to the equilibrium value. Moreover, the speed of convergence of the $w_{\mu\mu}$ is dictated by the sum rule $w_{12}+w_{11}=1$ (and similarly for $w_{22}$).

Finally, we remark that we do not observe the phenomenon of recurrence in our model. The reason is that the recurrence time is exponentially large in $N$ and in the theories we used to obtain our results, we have to take the limit $N\to\infty$ first such that the considered time $t$ is never of the order of a recurrence time.

\subsection{Three macro spaces}  \label{subsec:main3}
Let $\lambda>0$ be a small ($N$-independent) parameter, let $D\geq 1$ and let $N = (1+\lambda + \lambda^2)D$ be the dimension of the whole Hilbert space $\Hilbert$.\footnote{As in Section \ref{subsec:main2}, to avoid technicalities, we shall henceforth suppose that $N$ is in fact an integer.} We partition $\Hilbert$ into three macro spaces of dimension $d_1=\lambda^2 D$, $d_2 = \lambda D$, and $d_3 = D$, i.e., 
\begin{equation*}
	\Hilbert = \mathbb{C}^{N} \simeq \mathbb{C}^{d_1}\oplus\mathbb{C}^{d_2} \oplus \C^{d_3} =: \Hilbert_1 \oplus \Hilbert_2 \oplus \Hilbert_3\,. 
\end{equation*}
We model the Hamiltonian by a Hermitian random matrix $H=(h_{ij})_{i,j \in [N]}$ whose entries $h_{ij} = \overline{h_{ji}}$ satisfy the following conditions.

\begin{assumption}[$3 \times 3$ block matrix model] \label{ass:3}
For $i \le j$, the entries $h_{ij}$ are centered and independent random variables. The variance matrix $S := (\mathbb{E} |h_{ij}|^2)_{i,j}$ is given by 
\begin{align*}
	S=\left(\begin{matrix}
		S_{11} & S_{12} & 0\\
		S_{21} & S_{22} & S_{23}\\
		0 & S_{32} & S_{33}\\
	\end{matrix}\right),
\end{align*}
where
\footnotesize
\begin{alignat*}{2}
	S_{11} &= \frac{1}{(1+\lambda+\lambda^2) D} \frac{1}{\lambda^2} E_{d_1,d_1} \qquad 
	&&S_{12} = S_{21}^* = \frac{1+\lambda}{(1+\lambda+\lambda^2)D} E_{d_1,d_2}\\
	S_{22} &= \frac{1+\lambda-\lambda^3-\delta}{\lambda (1+\lambda+\lambda^2)D} E_{d_2,d_2}, \qquad 
	&&S_{23} = S_{32}^* = \frac{\delta}{(1+\lambda+\lambda^2)D} E_{d_2,d_3}\\
	& \qquad \qquad \qquad \qquad S_{33}= \frac{1+\lambda(1-\delta)+\lambda^2}{(1+\lambda+\lambda^2)D} && E_{d_3,d_3}
\end{alignat*}
\normalsize
and the zero blocks are of the appropriate size. Moreover, we assume that all centered moments of $\sqrt{N} H$ are uniformly bounded in $N$, i.e., for all $p\in\mathbb{N}$ there exists a ($\lambda$- and $\delta$-dependent) constant $C_p(\lambda, \delta)>0$ such that, uniformly in $1\leq i,j\leq N$,
\begin{align*}
	\mathbb{E}|\sqrt{N} h_{ij}|^p \leq C_p(\lambda, \delta). 
\end{align*} 
\end{assumption}

 We point out that, similarly to Assumption \ref{ass:2}, the variances in Assumption \ref{ass:3} are again chosen in such a way that the rows and columns of $S$ sum up to one. As in the $2\times 2$ case this choice ensures that the solution of the corresponding quadratic vector equation \cite{AEK19} is given by the semicirclular law. Again, we believe that the results obtained for this specific model remain true for more general matrices which have the same block structure as in Assumption~\ref{ass:3} with the sum of the variances in any secondary diagonal block being much smaller than in any diagonal block but not fine-tuned to ensure that $S$ is doubly stochastic. 
 
 In order to {further} simplify the {otherwise extremely tedious} computations {in Section \ref{sec:3proof}}, we choose $\delta=\lambda$ in the following. 

\subsubsection{Results for three macro spaces} For the $3\times 3$ block model defined above we have the following results. Their proofs are given in Section \ref{sec:3proof}. 
\begin{thm}[Normal typicality: $\dapp_{\mu \nu} \approx d_\nu/N$] \label{thm: dappmunu3} Let $H$ be a random matrix satisfying Assumption \ref{ass:3} with $\delta = \lambda$ and take $\mu, \nu \in \{1,2,3\}$. Denote 	\begin{equation} \label{eq:Mmunudef3}
		\dapp_{\mu\nu} = \frac{1}{d_\mu} \sum_{e \in \mathrm{spec}(H)} \tr(\Pi_e P_\mu \Pi_e P_\nu),
	\end{equation}
	where $\Pi_e$ is the projector on the eigenspace of $H$ corresponding to an eigenvalue $e \in \R$, and $P_\mu$ is the projector on $\Hilbert_\mu \subset \Hilbert$. Then it holds that
\begin{equation*}
	\dapp_{\mu \nu} = \frac{d_\nu}{N} + \mathcal{O}_\prec(C(\lambda)/\sqrt{N}) \,. 
\end{equation*}
\end{thm}

\begin{thm}[Dynamical typicality: The $w_{\mu \nu}'s$] \label{thm: wmunu3} Let $H$ be a random matrix satisfying Assumption \ref{ass:3} with $\delta = \lambda$ and take $\mu, \nu \in \{1,2,3\}$.
		Denote 
		\begin{equation*}
			w_{\mu\nu}(t) = \frac{1}{d_\mu} \tr\left[P_\mu \exp(\ii tH)P_\nu \exp(-\ii t H)\right]
		\end{equation*}
		for $t \ge 0$, where $P_\mu$ is the projector on $\Hilbert_\mu \subset \Hilbert$. 		
		Then, it holds that
		\begin{subequations} \label{eq:wmunus3}
\begin{align}
w_{12}(t)&=\frac{d_2}{N} + \frac{3}{\pi(\lambda t)^3}\left(1+o(1)\right)+ \mathcal{O}_\prec(C(\lambda, t)/N), \\ 
w_{13}(t)&=\frac{d_3}{N} - \frac{4}{\pi (\lambda t)^3}\left(1+o(1)\right)+ \mathcal{O}_\prec(C(\lambda, t)/N), \\ 
w_{23}(t)&=\frac{d_3}{N}-\frac{1}{\pi(\lambda t)^3}\left(1+o(1)\right)+ \mathcal{O}_\prec(C(\lambda, t)/N), 	
\end{align}
		\end{subequations}
		where $o(1)$ denotes a quantity vanishing in the limit $t \to \infty$, $\lambda  \to 0$, and $t \lambda \to \infty$, and $C(\lambda, t)$ is a constant depending only on its arguments. 
\end{thm}
Analogously to Theorem \ref{thm: wmunu}, having the results for the three off-diagonal $w_{\mu \nu}$'s provided in Theorem \ref{thm: wmunu3}  easily allows to provide expressions for the other $w_{\mu \nu}$'s as well (cf.~Lemma \ref{lem: rel wmunu3}). 
\begin{rmk}[Explicit form of the $w_{\mu \nu}$'s] \label{rmk:precise3}
Our proof of Theorem \ref{thm: wmunu3} actually yields explicit expressions of the
$w_{\mu \nu}$'s in \eqref{eq:wmunus3} up to an error term of order $1/N$. These expressions are of a similar principal form as \eqref{eq:w12precise} in Remark \ref{rmk:precise} involving a linear combination of infinite sums over Bessel functions. However, we refrain from stating them explicitly for brevity of the presentation. 
\end{rmk}

\begin{thm}[Approach to equilibrium]\label{thm: Pnupsi3}
			Let $H$ be a random matrix satisfying Assumption~\ref{ass:3} with $\delta = \lambda$ and take $\mu, \nu \in \{1,2,3\}$.
	Let  $\psi_0\in\mathbb{S}(\Hilbert_\mu)$ and denote $\psi_t := \ee^{- \ii t H} \psi_0$ for $t \ge 0$. Then, denoting the projector on the Hilbert space $\Hilbert_\nu \subset \Hilbert$ by $P_\nu$,  it holds that 
	\begin{subequations} \label{eq:Pnupsis3}
		\begin{align}
			\Vert P_2 \psi_t \Vert^2  &= \frac{d_2}{N} +\frac{3}{\pi(\lambda t)^3}\left(1+o(1)\right) +  \mathcal{O}_\prec\left(C(\lambda, t)/\sqrt{N} \right)  \quad \text{for} \quad \psi_0 \in \mathbb{S}(\Hilbert_1),\\
			\Vert P_3 \psi_t \Vert^2 &= \frac{d_3}{N} - \frac{4}{\pi (\lambda t)^3}\left(1+o(1)\right)+  \mathcal{O}_\prec\left(C(\lambda, t)/\sqrt{N} \right)  \quad \text{for} \quad \psi_0 \in \mathbb{S}(\Hilbert_1),
		\end{align}
as well as 
		\begin{align}
			\Vert P_1 \psi_t \Vert^2  &= \frac{d_1}{N} + \frac{3}{\pi\lambda^2 t^3}\left(1+o(1)\right)  + \mathcal{O}_\prec\left(C(\lambda, t)/\sqrt{N}\right)  \quad \text{for} \quad \psi_0 \in \mathbb{S}(\Hilbert_2),\\
			\Vert P_3 \psi_t \Vert^2 &= \frac{d_3}{N} -\frac{1}{\pi(\lambda t)^3}\left(1+o(1)\right)+ \mathcal{O}_\prec\left(C(\lambda, t)/\sqrt{N}\right)  \quad \text{for} \quad \psi_0 \in \mathbb{S}(\Hilbert_2),
		\end{align}
		 and 
				\begin{align}
				\Vert P_1 \psi_t \Vert^2  &= \frac{d_1}{N} - \frac{4}{\pi\lambda t^3}\left(1+o(1)\right) + \mathcal{O}_\prec\left(C(\lambda, t)/\sqrt{N}\right)  \quad \text{for} \quad \psi_0 \in \mathbb{S}(\Hilbert_3),\\
				\Vert P_2 \psi_t \Vert^2 &= \frac{d_2}{N} - \frac{1}{\pi \lambda^2 t^3}\left(1+o(1)\right) + \mathcal{O}_\prec\left(C(\lambda, t)/\sqrt{N}\right)  \quad \text{for} \quad \psi_0 \in \mathbb{S}(\Hilbert_3)\,. 
			\end{align}
	\end{subequations}
Here, $o(1)$ denotes a quantity vanishing in the limit $t \to \infty$, $\lambda  \to 0$, and $t \lambda \to \infty$, and $C(\lambda, t)$ is a constant depending only on its arguments. 
\end{thm}

The remaining cases in \eqref{eq:Pnupsis3} (i.e.~with $P_\mu$ for $\psi_0 \in \mathbb{S}(\Hilbert_\mu)$) can easily be obtained via the sum rule $\sum_{\nu} \Vert P_\nu \psi_t \Vert^2 = 1$. Moreover, our proof of Theorem \ref{thm: Pnupsi3} actually yields similar closed expressions for the $\Vert P_\nu \psi_t\Vert^2$'s as in Remark \ref{rmk:precise3}, but we refrain from stating them for brevity of the presentation.

\subsubsection{Discussion} 
 The initial state in Theorem~\ref{thm: Pnupsi3} is again purely in one subspace.
  By analyzing the signs in the formulas we just obtained,
   we see the following: 
   \begin{itemize}
\item[(i)]  If $\psi_0\in\mathbb{S}(\Hilbert_3)$, then the equilibrium values of $\|P_\nu\psi_t\|^2$ for $\nu\neq 3$ are approached from below; 
\item[(ii)]  if $\psi_0\in\mathbb{S}(\Hilbert_2)$, then the equilibrium value of $\|P_3\psi_t\|^2$ is approached from below and the one of $\|P_1\psi_t\|^2$ from above; 
\item[(iii)]  if $\psi_0\in\mathbb{S}(\Hilbert_1)$, then the equilibrium value of $\|P_2\psi_t\|^2$ is approached from above and the one of $\|P_3\psi_t\|^2$ from below. 
   \end{itemize}
   
   \begin{figure}[h]\label{fig:3blocks}
   	\centering
   	\includegraphics[height=9cm]{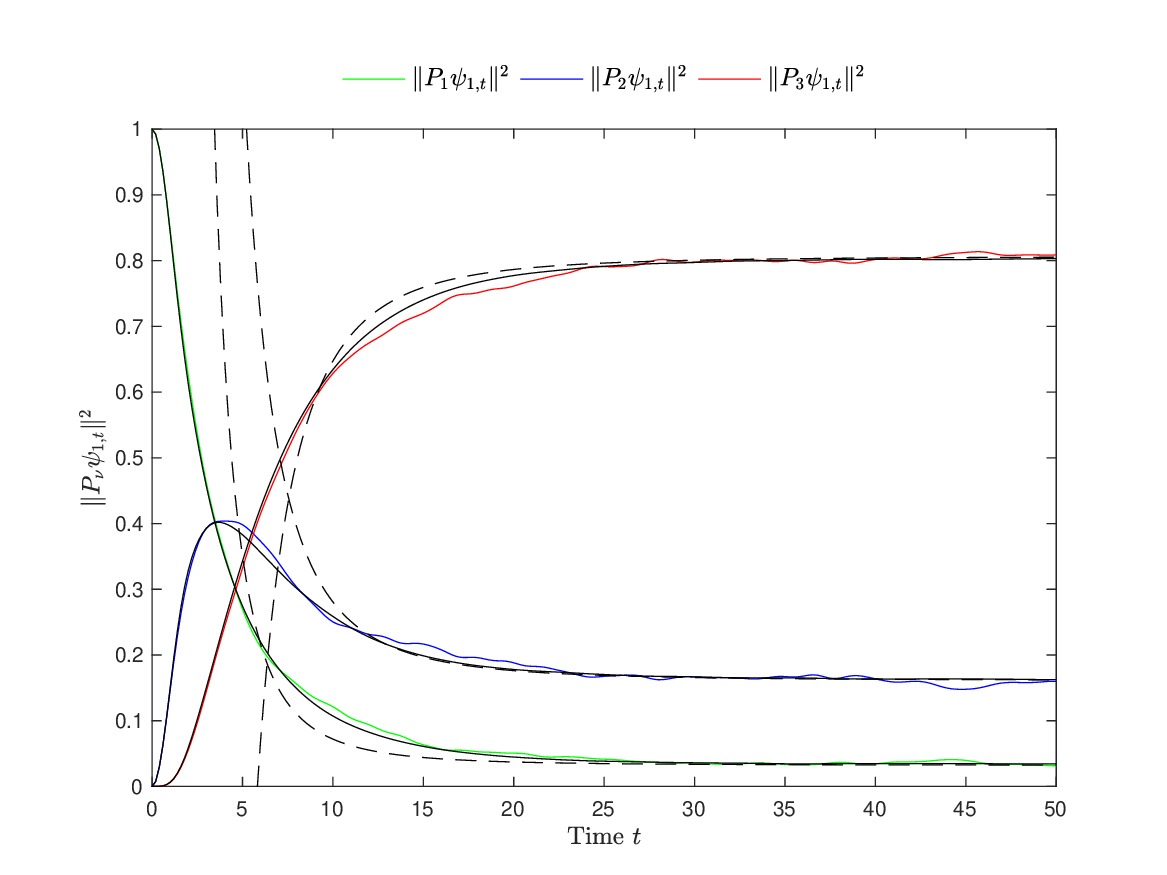}
   	\caption{Numerical simulation of the functions $t\mapsto\|P_\nu\psi_t\|^2$ for a random $3\times 3$ block matrix as in Assumption~\ref{ass:3}. Here, $\lambda=0.2$ and the Hilbert space $\Hilbert$ of dimension $N=4340$ is decomposed into 3 macro spaces of dimensions $d_1 = 140$ (green curve), $d_2=700$ (blue curve) and $d_3=3500$ (red curve). As in Figure~\ref{fig:2blocks}, the initial state $\psi_0\in\mathbb{S}(\Hilbert_1)$ was chosen purely randomly. The black solid curves are the deterministic {(w.r.t.~the randomly chosen $\psi_0$)} approximations $w_{1\nu}(t)$ and the black dashed curves are the approximations of $\|P_\nu\psi_t\|^2$ according to Theorem~\ref{thm: Pnupsi3}. Again, the dashed curves only start from $t\geq 1/(2 \lambda) = 2.5$.
   	A similar figure appeared as a purely numerical experiment in \cite{TTV23a, TTV23b} in the case of a random band matrix and four macro spaces. Here, although in a slightly different model, we rigorously prove the same qualitative behavior of the $t \mapsto \|P_\nu \psi_t\|$ curves.}
   \end{figure}
   
 More precisely, consider, e.g.,  the function $t\mapsto \|P_2\psi_t\|^2$ with $\psi_0\in\mathbb{S}(\Hilbert_1)$. It starts from zero, then jumps to a (positive) value of order one around $t\sim\lambda^{–1}$ (cf.~also the peak of the blue curve in Figure~\ref{fig:3blocks}),  and then relaxes to a value of order $\lambda$ for $t\gg \lambda^{-4/3}$. As another example, $t \mapsto \|P_3\psi_t\|^2$ approaches its equilibrium value already for $t\gg \lambda^{-1}$. That is, we observe equilibration on two different time scales. Note that the observed peak is due to our special choice of variances, $S_{13}=0$, which forbids a direct transition from $\Hilbert_1$ to $\Hilbert_3$. Analogous discussions can be made for the other curves.

Similarly to the case of a $2\times 2$ block matrix, we only wrote down the formulas $w_{12}, w_{13}$ and $w_{23}$ as all other $w_{\mu\nu}$ can be obtained from them with the help of Lemma~\ref{lem: rel wmunu3} and ``realistic'' equilibration times can be obtained by choosing the parameter $\lambda$ accordingly. We remark that the precise numbers in front of the error term $1/(\lambda t)^3$ in Theorem~\ref{thm: wmunu3} (and similarly in Theorem~\ref{thm: Pnupsi3}) do not carry any physical meaning.

\section{Two macro spaces: Proof of Theorems \ref{thm: dappmunu}--\ref{thm: Pnupsi}} \label{sec:2proof}
 In this section, we collect the proofs of our main results concerning the case of two macro spaces. 

\subsection{Proof of Theorem~\ref{thm: dappmunu}} \label{subsec:dappmunu2}
Let $u_1,\dots,u_N$ be the $\ell^2$-normalized eigenvectors of $H$ with corresponding eigenvalues $e_1\leq \dots \leq e_N$. Then we can write $\dapp_{\mu\nu}$ from \eqref{eq:Mmunudef} as
\begin{align*}
	\dapp_{\mu\nu} = \frac{1}{d_\mu} \sum_{e \in \mathrm{spec}(H) } \sum_{j,k\in I_e} \langle u_j|P_\mu|u_k\rangle\langle u_k|P_\nu|u_j\rangle,
\end{align*}
where the first sum is over all \emph{distinct} eigenvalues of $H$ and $I_e$ is the index set corresponding to a fixed energy $e$, i.e.~$I_e := \{ j \in [N] : H u_j = e u_j \}$. Therefore computing $\dapp_{\mu\nu}$ reduces to the computation of quantities of the form $\langle u_k|P_\nu|u_j\rangle$. For this computation we make use of the \emph{Eigenstate Thermalization Hypothesis} (ETH) for Wigner-type matrices, which provides a deterministic approximation to the quadratic forms $\braket{u_j | P_\nu | u_k}$, see \eqref{eq:ETH} below.  

In order to apply ETH for Wigner-type matrices, we have to check that the unique (see, e.g., \cite[Theorem~2.1]{AEK19}) solution $\mathbf{m} = (m_j)_{j=1}^N$ of the vector Dyson equation
\begin{align} \label{eq:VDE}
	-\frac{1}{m_j(z)} = z + \sum_{k=1}^{N} S_{jk} m_k(z), \quad \IM z\; \IM m_j(z)>0,
\end{align}
is bounded uniformly in $N$. Due to the structure of $S$ we obtain effectively only two equations: One for the first $d_1$ (identical) components of $\mathbf{m}$, all denoted by $m_1$, and one for the other $d_2$ (again identical) components, all denoted by $m_2$ (with a slight abuse of notation), i.e.
\begin{align*}
	-\frac{1}{m_1(z)} &= z + \frac{m_1(z)}{1+\lambda} + \frac{\lambda m_2(z)}{1+\lambda}\\
	-\frac{1}{m_2(z)} &= z + \frac{\lambda^2}{1+\lambda} m_1(z)  + \frac{1+\lambda-\lambda^2}{1+\lambda} m_2(z).
\end{align*}
We immediately see that $m_1(z)=m_2(z) = m_{\smc}(z)$, where 
\begin{equation*}
m_{\smc}(z) = \int_{\R} \frac{\rho_{\smc}(x) \dd x}{x-z} \quad \text{with} \quad \rho_{\smc}(x) := \frac{1}{2 \pi} \sqrt{[4-x^2]_+}
\end{equation*}
being the semicircular density of states, solves the equations on the upper half-plane and therefore is the unique solution of the quadratic vector equation. 

As the solution $\mathbf{m}(z)$ is obviously bounded in the $\|\cdot\|_\infty$-norm for each $z\in\mathbb{C}$ independently of $N$ and $S_{jk} \ge c(\lambda)/N$ for some constant $c(\lambda)$ depending only on its argument, uniformly in $j, k \in [N]$, we can apply the ETH for Wigner-type matrices \cite[Theorem 2.3]{ER24} and obtain\footnote{Strictly speaking, so far, ETH has only been proved in the \emph{bulk} of the spectrum (i.e.~for energies in $(-2,2)$). However, it is expected that it also holds for eigenstates corresponding to edge eigenvalues (cf.~\cite[Remark 4.1]{ER24}), which is the statement we use here.}
\begin{equation} \label{eq:ETH}
	\langle u_k|P_\nu|u_j\rangle = \delta_{kj}\frac{\langle\IM M(\gamma_j) P_\nu\rangle}{\pi\rho_{\smc}(\gamma_j)} + \mathcal{O}_\prec(C(\lambda)/\sqrt{N}) \,. 
\end{equation}
Here,  $C(\lambda)$ is a constant depending only on its argument, whose precise value might change from line to line. 

Moreover, here, $M(z) := \mbox{diag}(\mathbf{m}(z))$, $M(x) := M(x + \ii 0)$ for $x \in \R$, and $\gamma_j$ is the $j/N$-quantile of $\rho_{\smc}$, implicitly defined via
\begin{equation*}
	\int_{-\infty}^{\gamma_j} \rho_\smc(x)\, \dd x = \frac{j}{N}.
\end{equation*}
Since $\rho_\smc (\gamma_j) = \frac{1}{\pi}\, \IM m_\smc(\gamma_j)$
and $\IM M(\gamma_j) = \IM m_\smc(\gamma_j) \mathbf{1}$, we 
have
\begin{equation*}
	\langle u_k|P_\nu|u_j\rangle = \delta_{kj}\frac{d_\nu}{N} + \mathcal{O}_\prec(C(\lambda)/\sqrt{N}).
\end{equation*}
Let $\varepsilon>0$. By eigenvalue rigidity \cite[Corollary 1.11]{AEK17}, $|I_e|\leq N^\varepsilon$ with very high probability. This implies, with very high probability, that
\begin{equation}
	\begin{split}
		\dapp_{\mu\nu} &= \frac{1}{d_\mu} \sum_{e \in \mathrm{spec}(H) } \sum_{j,k\in I_e} \langle u_j|P_\mu| u_k\rangle\left(\delta_{kj} \frac{d_\nu}{N}+\mathcal{O}_\prec(C(\lambda)/\sqrt{N})\right)\\
		&= \frac{d_\nu}{N} + \frac{1}{d_\mu} \sum_{e \in \mathrm{spec}(H) } \sum_{j,k\in I_e} \langle u_j|P_\mu|u_k\rangle \mathcal{O}_\prec(C(\lambda)/\sqrt{N}) = \frac{d_\nu}{N} + \mathcal{O}_\prec(C(\lambda)/N^{1/2-\varepsilon}).
	\end{split} 
\end{equation}
	Since $\varepsilon>0$ was arbitrary, this finishes the proof. \qed 
\subsection{Proof of Theorem \ref{thm: wmunu}} \label{subsec:wmunupf}

We begin with the following lemma, providing a relation between the various $w_{\mu \nu}$'s. Its proof is given in Section \ref{subsubsec:pfwmunu}. 
\begin{lemma}[Relations between the $w_{\mu\nu}$]\label{lem: rel wmunu}
	Let $\mu, \nu \in \{1,2\}$ be two macro states and let $t\in\mathbb{R}$. Then,
	\begin{equation*}
		w_{\mu\nu}(t) = \frac{d_\nu}{d_\mu} w_{\nu\mu}(-t) \quad \text{and} \quad 
		\sum_{\nu'} w_{\mu\nu'}(t) = 1. 
	\end{equation*}
\end{lemma}
Hence, as already mentioned in Section \ref{sec:main}, we only have to compute $w_{12}(t)$ as with the aid of Lemma~\ref{lem: rel wmunu} all the other $w_{\mu\nu}$ can be obtained from this. For $w_{12}(t)$, we express the time evolutions as contour integrals of resolvents $G(z) = (H-z)^{-1}$:  
\begin{equation}
	\begin{split}
		w_{12}(t) &= \frac{1}{\lambda D} \tr\left[P_1 \exp(\ii tH) P_2 \exp(-\ii tH)\right]\\
		&= \frac{1}{\lambda D} \frac{1}{(2\pi \ii)^2} \oint_{\gamma}\oint_{\tilde{\gamma}} \ee^{\ii t(z-\tilde{z})} \tr(P_1 G(z) P_2 G(\tilde{z}))\, \dd z\, \dd \tilde{z}\\
		&= \frac{1+\lambda}{\lambda} \frac{1}{(2\pi \ii)^2} \oint_{\gamma}\oint_{\tilde{\gamma}} \ee^{\ii t(z-\tilde{z})} \langle P_1 G(z) P_2 G(\tilde{z})\rangle\, \dd z \, \dd \tilde{z} \label{eq: w12 PGPG}
	\end{split}
\end{equation}
where we choose the contours as
\begin{equation} \label{eq:contourchoice}
\gamma = \tilde{\gamma} = \{ z \in \C : \mathrm{dist}(z, [-3,3]) = t^{-1} \} \,. 
\end{equation}
Since $\mathrm{spec}(H) \subset [-2-\epsilon, 2+ \epsilon]$ with very high probability (as a simple consequence of eigenvalue rigidity \cite[Corollary 1.11]{AEK17}), for every $\epsilon > 0$, this choice of contours ensures the representation \eqref{eq: w12 PGPG} to be valid, again with very high probability.

In order to evaluate \eqref{eq: w12 PGPG}, we rely on the following \emph{two-resolvent global law}, i.e.~a concentration estimate for the product of resolvents and deterministic matrices appearing in \eqref{eq: w12 PGPG}. In fact, completely analogously to \cite[Proposition~3.2]{ER24} (see, in particular its proof in \cite[Appendix~A]{ER24}), we find that 
\begin{equation} \label{eq:globallaw}
	\langle P_1 G(z) P_2 G(\tilde{z})\rangle = \langle P_1 M(z,P_2,\tilde{z}) \rangle  + \mathcal{O}_\prec(C(\lambda, t)/N),
\end{equation}
uniformly in spectral parameters $z \in \gamma,\tilde{z} \in \tilde{\gamma}$. The deterministic approximation  to the random resolvents in \eqref{eq:globallaw} is given in terms of $M(z,P_2,\tilde{z})$, which is defined by
\begin{align} \label{eq:M2def}
	M(z,P_2,\tilde{z}) = m_\smc(z)m_\smc(\tilde{z})\mbox{diag}\left(\mathcal{B}^{-1}_{z,\tilde{z}}P_2^{\mathrm{diag}}\right),
\end{align}
where $P_2^{\mathrm{diag}} := (P_{2,jj})_{j=1}^N$. The operator 
\begin{equation}
	\mathcal{B}_{z,\tilde{z}} := \mathbf{1}-m_\smc(z)m_\smc(\tilde{z}) S
\end{equation}
acting on vectors in $\C^N$ is called the \emph{two-body stability operator}, whose inverse can be explicitly computed. The proof of the following Lemma \ref{lem: inverse stab op} is a direct computation and hence omitted. 

\begin{lemma}[Inverse of the stability operator]\label{lem: inverse stab op}
The inverse of the two-body stability operator, $\mathcal{B}_{z,\tilde{z}}^{-1}$, at spectral paramters $z, \tilde{z} \in \C$, is given by
\footnotesize
	\begin{align*}
		\mathcal{B}_{z,\tilde{z}}^{-1} =  
		\mathbf{1} + \frac{\mathfrak{m}}{D(1+\lambda)(1-\mathfrak{m})(1-(1-\lambda)\mathfrak{m})} \left(\begin{matrix}
			\lambda^{-1}(1-\mathfrak{m}+\lambda^2\mathfrak{m}) E_{d_1,d_1} & \lambda E_{d_1,d_2}\\
			\lambda E_{d_2,d_1} & (1+\lambda-\lambda^2- \mathfrak{m}+\lambda^2 \mathfrak{m}) E_{d_2,d_2}\\
		\end{matrix}\right),
	\end{align*}
	\normalsize
	where $\mathfrak{m}=m_\smc(z) m_\smc(\tilde{z})$.
\end{lemma}

In this way, we find that \footnotesize
\begin{align}
	M(z,P_2,\tilde{z}) = \mathfrak{m} \left(\begin{matrix}
		0&0\\
		0& I_{d_2,d_2}\\
	\end{matrix}\right)	+ \frac{\mathfrak{m}^2}{(\mathfrak{m}-1)(1+\lambda)[(1-\lambda)\mathfrak{m}-1]} \left(\begin{matrix}
		\lambda I_{d_1,d_1} & 0\\
		0 & (1+\lambda-\lambda^2-\mathfrak{m}+\lambda^2 \mathfrak{m}) I_{d_2,d_2}
	\end{matrix}\right)\label{eq: MzP2z}
\end{align}
\normalsize
and therefore
\begin{align*}
	\langle P_1 G(z) P_2 G(\tilde{z})\rangle = \frac{1}{1+\lambda} \frac{\lambda^2 \mathfrak{m}^2}{(\mathfrak{m}-1)(1+\lambda)[(1-\lambda)\mathfrak{m}-1]} + \mathcal{O}_\prec(C(\lambda, t)/N).
\end{align*}
Plugging this into \eqref{eq: w12 PGPG} and using that $|\ee^{\ii t(z-\tilde{z})}| \lesssim 1$ for $z \in \gamma$ and $\tilde{z} \in \tilde{\gamma}$, leads to
\begin{align*}
	w_{12}(t) = \frac{1}{(2\pi \ii)^2} \oint_{\gamma}\oint_{\tilde{\gamma}} \ee^{\ii t(z-\tilde{z})} \frac{\lambda \mathfrak{m}^2}{(\mathfrak{m}-1)(1+\lambda)[(1-\lambda)\mathfrak{m}-1]}\, \dd z\, \dd \tilde{z} + \mathcal{O}_\prec(C(\lambda, t)/N).
\end{align*}

To evaluate this, we make use of the following lemma, whose proof is given in Section \ref{subsubsec:pfdenomlemma}. 

\begin{lemma}[Contour integrals of $m_\smc$]\label{lem: integral fraction}
	Let $r\in(0,1)$ and let $\mathfrak{m}$ be as in Lemma~\ref{lem: inverse stab op}. Then,
	\begin{equation}
		\begin{split}
		&\frac{1}{(2\pi \ii)^2}\oint_\gamma \oint_{\tilde{\gamma}} \ee^{\ii t(z-\tilde{z})}\frac{\mathfrak{m}^2}{(1-\mathfrak{m})(1-r\mathfrak{m})}\, \dd z\, \dd \tilde{z} \\
		=  \,  &\frac{1}{1-r} \left( 1 - \sum_{n \ge 0} r^n (n+1)^2 \left(\frac{J_{n+1}(2t)}{t}\right)^2 \right) =
\frac{1}{1-r} - \frac{1}{\pi t^3 (1-r)^4}(1 + o(1)), 
		\end{split}
	\end{equation}
	where $o(1)$ vanishes in the limit $t \to \infty$, $r \uparrow 1$, and $t (1-r) \to \infty$. Moreover, $J_n$ denotes the $n$-th order Bessel function of the first kind. 
\end{lemma}

An application of Lemma~\ref{lem: integral fraction} with $r=1-\lambda$, and recalling that $d_2/N = 1/(1 + \lambda)$, immediately shows that
\begin{equation*}
	w_{12}(t) = \frac{d_2}{N} - \frac{1}{\pi(\lambda t)^3} \left(1+o(1)\right) + \mathcal{O}_\prec(C(\lambda, t)/N).
\end{equation*}
The explicit form provided in Remark \ref{rmk:precise} is also immediate from Lemma \ref{lem: integral fraction}. Finally, as mentioned above, by Lemma~\ref{lem: rel wmunu} and using the symmetry $t \to -t$ (since $-H$ is again a matrix satisfying Assumption \ref{ass:2}), the other $w_{\mu\nu}$ can readily be obtained from $w_{12}$. This completes the proof of Theorem \ref{thm: wmunu}. \qed

We are now left with giving the proofs of Lemmas \ref{lem: rel wmunu} and \ref{lem: integral fraction}.

\subsubsection{Proof of Lemma \ref{lem: rel wmunu}} \label{subsubsec:pfwmunu}
Using the definition $w_{\mu\nu}(t) = d_\mu^{-1} \tr\left[P_\mu \exp(\ii tH)P_\nu \exp(-\ii t H)\right]$, one can easily see that $	\sum_{\nu'} w_{\mu\nu'}(t) = 1$ since $\sum_{\nu'} P_{\nu'}=1$.

Using the definition again, we find
\begin{equation*}
	w_{\mu\nu}(t) = \frac{1}{d_\mu} \tr\left(P_\mu \ee^{\ii t H} P_\nu \ee^{-\ii t H}\right) = \frac{d_\nu}{d_\mu} \frac{1}{d_\nu} \tr\left(P_\nu \ee^{-\ii tH} P_\mu \ee^{\ii t H}\right) = \frac{d_\nu}{d_\mu} w_{\nu\mu}(-t).
\end{equation*}
This concludes the proof. \qed
\subsubsection{Proof of Lemma \ref{lem: integral fraction}} \label{subsubsec:pfdenomlemma}
Since $|\mathfrak{m}| < 1$ for $z \in \gamma$ and $\tilde{z} \in \tilde{\gamma}$, we can write 
	\begin{equation}
		\frac{\mathfrak{m}^2}{(1-\mathfrak{m})(1-r\mathfrak{m})} = \frac{1}{1-r}\left(\sum_{n\geq 0} \mathfrak{m}^{n+2} - r^{-1}\sum_{k\geq 0} (r\mathfrak{m})^{k+2}\right).
	\end{equation}
Therefore we find that
\begin{align} \label{eq:T1T2}
	\frac{1}{(2\pi \ii)^2} \oint_{\gamma}\oint_{\tilde{\gamma}} \ee^{\ii t(z-\tilde{z})} \frac{\mathfrak{m}^2}{(1-\mathfrak{m})(1-r\mathfrak{m})}\, \dd z\, \dd \tilde{z} = \frac{1}{1-r}(T_1-T_2),
\end{align}
where
\begin{align}
\label{eq:T1def}	T_1 &:= \frac{1}{(2\pi \ii)^2} \sum_{n\geq 0}  \oint_{\gamma}\oint_{\tilde{\gamma}} \ee^{\ii t(z-\tilde{z})} (m\tilde{m})^{n+2}\, \dd z\, \dd\tilde{z},\\
\label{eq:T2def}	T_2 &:= \frac{1}{(2\pi \ii)^2} r^{-1} \sum_{n\geq 0} \oint_{\gamma}\oint_{\tilde{\gamma}} \ee^{\ii t(z-\tilde{z})} (rm\tilde{m})^{n+2}\, \dd z\, \dd \tilde{z}.
\end{align}
Here we introduced the shorthand notation $m=m_{\smc}(z)$ and $\tilde{m} = m_{\smc}(\tilde{z})$. These two contributions $T_1, T_2$ shall now be computed separately.

We start with computing $T_1$. The integrals decouple and we get
\begin{equation*}
	\sum_{n\geq 0} \oint_{\gamma}\oint_{\tilde{\gamma}} \ee^{\ii t(z-\tilde{z})} (m\tilde{m})^{n+2}\, \dd z\, \dd \tilde{z} = \sum_{n\geq 0} \oint_{\gamma} \ee^{\ii tz} m^{n+2}\, \dd z \oint_{\tilde{\gamma}}\ee^{-\ii t\tilde{z}} \tilde{m}^{n+2}\, \dd \tilde{z}.
\end{equation*}

These can now be evaluated with the aid of the following two lemmas, whose proofs are given in Appendix \ref{app:aux}.

\begin{lemma}[Contour Integral of $m_\smc^n$]\label{lem: contour int}
	Let $n\in\mathbb{N}$ and $\gamma \subset \C$ a smooth contour encircling $[-2,2]$ once counterclockwise. Then,
	\begin{equation}
		\frac{1}{2\pi \ii}\oint_{\gamma} \ee^{\ii tz} m_\smc(z)^n \, \dd z = \ii^{n+1} J_{n+1}(-2t)- \ii^{n-1}J_{n-1}(-2t),
	\end{equation}
	where $J_k$ is the $k$th Bessel function of the first kind.
\end{lemma}

\begin{lemma}[Sums over products of $J_k$]\label{lem: sum prod J_k}
	Let $x,y\in\mathbb{C}$, $p\in\mathbb{N}$, $q\in 2\mathbb{N}$. Then,
	\begin{align}
		\sum_{n\geq 0} \ii^{-2n} J_n(x) J_n(y) &= \frac{1}{2}\left(J_0(x) J_0(y) + J_0(|x+y|)\right), \label{eq:sumprod1}\\
		\sum_{n\geq 0} \ii^{-2(n+p)} J_{n+p}(x) J_{n+p}(-x) &= \frac{1}{2}\left(1-J_0(x)^2\right) - \sum_{k=1}^{p-1} J_k(x)^2, \label{eq:sumprod2}\\
		\sum_{n\geq 0} \ii^{-(2n+q)} J_{n+q}(x) J_n(-x) &= \frac{1}{2}\left(i^{-q} J_q(x) J_0(x) - \sum_{k=0}^{q-1}\ii^{-(2k-q)}J_k(x) J_{k-q}(-x)\right). \label{eq:subprod3}
	\end{align}
\end{lemma}

With Lemma~\ref{lem: contour int} we obtain
\begin{align*}
	T_1 = \sum_{n\geq 0} \left[\ii^{n+3}J_{n+3}(-2t)-i^{n+1}J_{n+1}(-2t)\right]\left[\ii^{n+3}J_{n+3}(2t)-\ii^{n+1}J_{n+1}(2t)\right].
\end{align*}
An application of Lemma~\ref{lem: sum prod J_k} immediately shows that
\begin{align*}
	\sum_{n\geq 0} \ii^{2(n+1)} J_{n+1}(-2t) J_{n+1}(2t) &= \frac{1}{2}\left(1-J_0(2t)^2\right),\\
	\sum_{n\geq 0} \ii^{2(n+3)} J_{n+3}(-2t) J_{n+3}(2t)&= \frac{1}{2}\left(1-J_0(2t)^2\right) - J_1(2t)^2 - J_2(2t)^2.
\end{align*}
Moreover, we find that
	\begin{align*}
		\sum_{n\geq 0} \ii^{2n+4} &J_{n+3}(2t) J_{n+1}(-2t) = -\frac{1}{2} J_1(2t)^2 + J_0(2t) J_2(2t)
	\end{align*}
and similarly 
\begin{align*}
	\sum_{n\geq 0} \ii^{2n+4} J_{n+3}(-2t) J_{n+1}(2t) =  -\frac{1}{2} J_1(2t)^2 + J_0(2t) J_2(2t).
\end{align*}
Putting everything together we get that the $T_1$-contribution to \eqref{eq:T1T2} with $T_1$ from \eqref{eq:T1def} is given by
\begin{equation} \label{eq:T1final}
		\frac{1}{1-r} T_1 = \frac{1}{1-r} \left(1-\frac{J_1(2t)^2}{t^2}\right),
\end{equation}
where we used that $J_{n-1}(z)+J_{n+1}(z) = (2n/z) J_n(z)$, see, e.g., \cite[9.1.27]{AS72}.

Next we turn to the computation of $T_2$ from \eqref{eq:T2def}. With Lemma~\ref{lem: contour int} and using again that $J_{n-1}(z)+J_{n+1}(z) = (2n/z) J_n(z)$ we obtain
	\begin{align} \label{eq:T2start}
		T_2 = \sum_{n\geq 0} r^{n+1} \left(\frac{(n+2)J_{n+2}(2t)}{t}\right)^2.
	\end{align}
Writing the Bessel functions via their integral representation \cite[Eq. 8.411 1.]{Gradshteyn.Ryzhik.2007} as
	\begin{equation} \label{eq:intrep}
		J_{n+2}(2t)^2 = \frac{1}{(2\pi)^2} \int_{-\pi}^\pi \int_{-\pi}^\pi \ee^{\ii n(\theta+\theta')} \ee^{2\ii(\theta+\theta')} \ee^{-2\ii t(\sin\theta+\sin\theta')}\, \dd \theta\, \dd \theta'.
	\end{equation}
we find that 
	\begin{equation} \label{eq:statphasestart}
		\begin{split}
		&\sum_{n\geq 0} r^{n+1} \left(\frac{(n+2)J_{n+2}(2t)}{t}\right)^2  \\
		= &\frac{1}{(2\pi t)^2 r} \int_{-\pi}^\pi \int_{-\pi}^\pi \left(\frac{r \ee^{\ii(\theta+\theta')}(1+r\ee^{\ii(\theta+\theta')})}{(1-r \ee^{\ii(\theta+\theta')})^3}-r\ee^{\ii(\theta+\theta')}\right) \ee^{-2\ii t(\sin\theta+\sin\theta')}\, \dd \theta\, \dd \theta',
		\end{split}
	\end{equation}
where we additionally used that $\sum_{n\geq 0} n^2 y^n = y(y+1)/(1-y)^3$ for $|y| < 1$. 

We evaluate the integral with the help of the stationary phase approximation. To this end recall the standard form of the stationary phase lemma for smooth functions $f$ with non-degenerate critical points and smooth and compactly supported $g$: 
\begin{align}
	\int_{\mathbb{R}^2} g(x) \ee^{\ii tf(x)}\, \dd x = \frac{2 \pi}{t}\sum_{x_0} \frac{\ee^{\ii tf(x_0)}  \ee^{\frac{\ii\pi}{4}\mathrm{sgn}(\mathrm{Hess}(f(x_0)))} }{\sqrt{|\det\mathrm{Hess}(f(x_0))|}} g(x_0) + o(t^{-1}), \label{eq: stat phase approx}
\end{align}
as $t\to\infty$, where the error term implicitly depends on the derivatives of $f$ and $g$ near the critical points $x_0$. The sum in \eqref{eq: stat phase approx} ranges over all critical points of $f$ and $\mbox{sgn}(\mbox{Hess}(f(x_0)))$ denotes the signature of the Hessian of $f$ at $x_0$, i.e., the number of positive minus the number of negative eigenvalues.

We aim to apply \eqref{eq: stat phase approx} with the functions\footnote{We shall henceforth ignore that our integral in \eqref{eq: stat phase approx} is only over $[-\pi, \pi]^2$, since this can easily be accommodated by a suitable cutoff function. }
\begin{align}
	f(\theta,\theta') &:= -2(\sin\theta+\sin\theta'),\\
	g(\theta,\theta')&:= \frac{r \ee^{\ii(\theta+\theta')}(1+r\ee^{\ii(\theta+\theta')})}{(1-r\ee^{\ii(\theta+\theta')})^3}-r\ee^{\ii(\theta+\theta')}.
\end{align}
The function $f$ has four non-degenerate critical points at
\begin{alignat*}{2}
x_1 &= (\pi/2, \pi/2) \,, \qquad &&x_2 = (-\pi/2, - \pi/2) \,,  \\
x_3 &= (\pi/2, -\pi/2) \,, \qquad &&x_4 = (-\pi/2, \pi/2) \,
\end{alignat*}
for which we have $\sqrt{|\det \mbox{Hess}(f(x_i))|}=2$ for all $i=1,\dots,4$ and the signature of the Hessian is given by 
\begin{alignat*}{2}
	\mathrm{sgn}(\mathrm{Hess}(f(x_1))) &= 2 \,, \qquad &&	\mathrm{sgn}(\mathrm{Hess}(f(x_2))) = -2\,,  \\
		\mathrm{sgn}(\mathrm{Hess}(f(x_3)))&= 0\,, \qquad &&	\mathrm{sgn}(\mathrm{Hess}(f(x_4)) )= 0\,.
\end{alignat*}
At the critical points $x_1, x_2$, the stationary phase lemma can be applied without any further difficulty, since, near these points, $g$ is smooth and bounded function, uniformly as $r \uparrow 1$. By simple computation, the contribution of $x_1, x_2$ is thus given by 
\begin{equation} \label{eq:easystat}
\frac{2 \pi r}{t}\left(1 - \frac{(1-r)}{(1+r)^3} \right) \sin (4t) + o(t^{-1}) \quad \text{for} \quad t \to \infty \,, 
\end{equation}
where the implicit constant in the error term is independent of $r$ and $t$. 

However, near the other two critical points, $x_3, x_4$, the function $g$ is singular as $1/(1-r)^3 \gg 1$ (recall that $r \uparrow 1$). Moreover, with each derivative, the size of $g$ gets enhanced by another factor of $(1-r)^{-1}$ close to $x_3, x_4$. Hence, integration by parts away from the stationary points yields error terms inverse polynomially in $t(1-r)$. 

We now check that $t(1-r)$ is also the effective stationary phase parameter. Indeed, in local coordinates $(x,y)$ near the stationary phase points $x_3$ or $x_4$, by Taylor-expanding both $f$ and $g$, the dominant contribution of the integral \eqref{eq:statphasestart} can be written as
\begin{equation*}
\int_B \ee^{\ii t (x^2 - y^2)} \frac{1}{((1-r) - \ii (x+y))^3} \dd x \dd y
\end{equation*}
for an appropriate $O(1)$ neighborhood $B$ of the origin in $\R^2$. Note that the signature of the Hessian of $(x,y) \mapsto x^2 - y^2$ is negative, and the integrand depends only on the sum $x+y$. Hence, by changing variables to $u := x+y$ and $v := x-y$, we find the dominant contribution to take the form
\begin{equation} \label{eq:statphase}
\int_{\tilde{B}} \ee^{\ii t uv} \frac{1}{((1-r) - \ii u)^3} \dd u \dd v
\end{equation}
for an appropriately transformed version $\tilde{B}$ of $B$. Finally, a simple rescaling $u \to u/(1-r)$, which regularizes the denominator in \eqref{eq:statphase}, shows that $t (1-r)$ arises as the effective stationary phase parameter.

In conclusion, for the stationary phase approximation near these points to be effective -- in the sense that the error term in \eqref{eq: stat phase approx} is actually smaller than the leading contribution from the stationary points, we need that $t (1-r) \gg 1$. In this regime, following through how the $o(1/t)$ error in \eqref{eq: stat phase approx} depends on the singular behavior of $g$, we computed the contribution of the stationary points $x_3, x_4$ as 
\begin{equation} \label{eq:hardstat}
\frac{2 \pi}{t} \left( \frac{r(1+r)}{(1-r)^3}-r\right) + o((1-r)^{-3}t^{-1})  \quad \text{for} \quad t(1-r) \to \infty \,, 
\end{equation}
where again the implicit constant in the error term is independent of $r$ and $t$. 

Collecting the terms from \eqref{eq:easystat}--\eqref{eq:hardstat} we find that 
	\begin{equation} \label{eq:T2}
	T_2 = 	\sum_{n\geq 0} r^{n+1}\left(\frac{(n+2)J_{n+2}(2t)}{t}\right)^2 = \frac{1}{\pi}\frac{1}{((1-r)t)^3} (1 +o(1))
	\end{equation}
	where $o(1)$ vanishes in the limit $t \to \infty$, $r \uparrow 1$, and $t (1-r) \to \infty$. Here, we in particular used that the contributions from $x_1, x_2$ in \eqref{eq:easystat} are much smaller than the contribution from $x_3, x_4$ in \eqref{eq:hardstat}, in the sense of \eqref{eq:T2}. 
Combining \eqref{eq:T1final} with \eqref{eq:T2start} and \eqref{eq:T2}, we arrive at
\begin{equation*}
	\frac{1}{1-r} (T_1-T_2)= \frac{1}{1-r} \left( 1 - \sum_{n \ge 0} r^n (n+1)^2 \left(\frac{J_{n+1}(2t)}{t}\right)^2 \right)  = \frac{1}{1-r}-\frac{1}{\pi}\frac{1}{(1-r)^4t^3} (1 +o(1)) \,, 
\end{equation*}
where $o(1)$ vanishes in the limit $t \to \infty$, $r \uparrow 1$, and $t (1-r) \to \infty$. This finishes the proof. \qed

\subsection{Proof of Theorem~\ref{thm: Pnupsi}} 
Similarly as in the proof of Theorem~\ref{thm: wmunu} we find for $\psi_0\in\mathbb{S}(\Hilbert_\mu)$ that
	\begin{align*}
		\|P_\nu\psi_t\|^2  &= \langle \psi_0| \exp(\ii tH) P_\nu \exp(-\ii tH)|\psi_0 \rangle\\
		&= \frac{1}{(2\pi \ii)^2} \oint_{\gamma}\oint_{\tilde{\gamma}} \ee^{\ii t(z-\tilde{z})} \langle \psi_0 | G(z) P_\nu G(\tilde{z}) | \psi_0 \rangle\, \dd z\, \dd \tilde{z}
	\end{align*}
	where the contours are chosen as in \eqref{eq:contourchoice}. Then, the proof proceeds just as in Section \ref{subsec:wmunupf}, except that the \emph{average} two-resolvent global law in \eqref{eq:globallaw} is replaced by the \emph{isotropic} two-resolvent global law\footnote{This can be obtained completely analogously to \cite[Proposition~3.2]{ER24} (see, in particular its proof in \cite[Appendix~A]{ER24}).}
	\begin{equation*}
\langle \psi_0| G(z) P_\nu G(\tilde{z})|\psi_0 \rangle = \langle \psi_0 | M(z, P_\nu, \tilde{z}) | \psi_0 \rangle + \mathcal{O}_\prec(C(\lambda, t)/\sqrt{N}) \,. 
	\end{equation*}
	Since $M(z, P_\nu, \tilde{z}) $ from \eqref{eq:M2def} is a diagonal matrix and constant within each block, the rest of the argument works exactly the same as in the proof of Theorem \ref{thm: dappmunu}. \qed

\section{Three macro spaces: Proof of Theorems \ref{thm: dappmunu3}, \ref{thm: wmunu3}, and \ref{thm: Pnupsi3}} \label{sec:3proof}
In this section, we collect the proofs of our main results concerning the case of three macro spaces. 
\subsection{Proof of Theorem~\ref{thm: dappmunu3}}
Let $u_1,\dots,u_N$ be the $\ell^2$-normalized eigenvectors of $H$ with corresponding eigenvalues $e_1\leq \dots \leq e_N$. Then, as in the proof of Theorem~\ref{thm: dappmunu}, we can write $\dapp_{\mu\nu}$ as
\begin{align}
	\dapp_{\mu\nu} = \frac{1}{d_\mu} \sum_{e \in \mathrm{spec}(H)} \sum_{j,k\in I_e} \langle u_j|P_\mu|u_k\rangle\langle u_k|P_\nu|u_k\rangle,
\end{align}
where the first sum is over all \emph{distinct} eigenvalues of $H$ and $I_e$ is the index set corresponding to a fixed energy $e$, i.e.~$I_e := \{ j \in [N] : H u_j = e u_j \}$.
Again, we aim to apply the ETH for Wigner-type matrices \cite[Theorem~2.3]{ER24} and therefore check that all conditions are satisfied.

First, note that, similarly to Section \ref{subsec:dappmunu2}, the unique solution to the vector Dyson equation \eqref{eq:VDE} (but now with $S$ from Assumption \ref{ass:3}), is given by $\mathbf{m}(z) = (m_\smc(z), ..., m_\smc(z))$, i.e.~the vector having all its entries equal to $m_\smc(z)$. 
This vector is obviously bounded in the $\|\cdot\|_\infty$-norm for each $z\in\mathbb{C}$ independently of $N$.

Second, the variance matrix $S$ is uniformly primitive, i.e., there exists an integer $L$ and an $N$-independent constant $c(\lambda)>0$ such that
\begin{align*}
	(S^L)_{jk} \geq \frac{c(\lambda)}{N} \,. 
\end{align*}
A direct computation shows that this is fulfilled for $L=2$. All other conditions of Theorem~2.3 in \cite{ER24} are obviously fulfilled in the present model. Therefore we can apply the ETH for Wigner-type matrices and the result follows in exactly the same way as in the proof of Theorem~\ref{thm: dappmunu}. \qed

\subsection{Proof of Theorem \ref{thm: wmunu3}} 
Just as in the proof of Theorem \ref{thm: wmunu}, we begin with the following lemma, providing a relation between the various $w_{\mu \nu}$'s. Its proof is identical to the one of Lemma \ref{lem: rel wmunu} and so omitted. 
\begin{lemma}[Relations between the $w_{\mu\nu}$]\label{lem: rel wmunu3}
	Let $\mu, \nu \in \{1,2,3\}$ be two macro states and let $t\in\mathbb{R}$. Then,
	\begin{equation*}
		w_{\mu\nu}(t) = \frac{d_\nu}{d_\mu} w_{\nu\mu}(-t) \quad \text{and} \quad 
		\sum_{\nu'} w_{\mu\nu'}(t) = 1. 
	\end{equation*}
\end{lemma}
Hence, we only have to compute three of the $w_{\mu \nu}$'s (the rest can be deduced with the aid of Lemma \ref{lem: rel wmunu3}) and we choose 
\begin{equation} \label{eq:choiceofws}
w_{12}(t)\,, \qquad w_{13}(t)\,, \qquad \text{and} \qquad w_{23}(t) \,. 
\end{equation}
For a general $w_{\mu \nu}(t)$ from \eqref{eq:choiceofws}, we express the time evolutions as contour integrals of resolvents $G(z) = (H-z)^{-1}$:  
\begin{equation}
	\begin{split}
		w_{\mu \nu}(t) &= \frac{1}{d_\mu} \tr\left[P_\mu \exp(\ii tH) P_\nu \exp(-\ii tH)\right]\\
		&= \frac{1}{d_\mu} \frac{1}{(2\pi \ii)^2} \oint_{\gamma}\oint_{\tilde{\gamma}} \ee^{\ii t(z-\tilde{z})} \tr(P_\mu G(z) P_\nu G(\tilde{z}))\, \dd z\, \dd \tilde{z}\\
		&= \frac{N}{d_\mu} \frac{1}{(2\pi \ii)^2} \oint_{\gamma}\oint_{\tilde{\gamma}} \ee^{\ii t(z-\tilde{z})} \langle P_\mu G(z) P_\nu G(\tilde{z})\rangle\, \dd z \, \dd \tilde{z} \label{eq: wmunu PGPG}
	\end{split}
\end{equation}
where we choose the contours $\gamma, \tilde{\gamma}$ as in \eqref{eq:contourchoice}. 
Since $\mathrm{spec}(H) \subset [-2-\epsilon, 2+ \epsilon]$ with very high probability (as a simple consequence of eigenvalue rigidity \cite[Corollary 1.11]{AEK17}), for every $\epsilon > 0$, this choice of contours ensures the representation \eqref{eq: wmunu PGPG} to be valid, again with very high probability.

Analogously to the proof of Theorem \ref{thm: wmunu}, in order to evaluate \eqref{eq: wmunu PGPG}, we rely on the following \emph{two-resolvent global law} 
\begin{equation} \label{eq:globallaw3}
	\langle P_\mu G(z) P_\nu G(\tilde{z})\rangle = \langle P_\mu M(z,P_\nu,\tilde{z}) \rangle  + \mathcal{O}_\prec(C(\lambda, t)/N),
\end{equation}
uniformly in spectral parameters $z \in \gamma,\tilde{z} \in \tilde{\gamma}$. Just as \eqref{eq:globallaw}, this can be obtained completely analogously to \cite[Proposition~3.2]{ER24} (see, in particular its proof in \cite[Appendix~A]{ER24}).

The deterministic approximation  to the random resolvents in \eqref{eq:globallaw} is obtained in the exact same way as in \eqref{eq:M2def}, i.e.~given by
\begin{align} \label{eq:M2def3}
	M(z,P_\nu,\tilde{z}) = m_\smc(z)m_\smc(\tilde{z})\mbox{diag}\left(\mathcal{B}^{-1}_{z,\tilde{z}}P_\nu^{\mathrm{diag}}\right),
\end{align}
where $P_\nu^{\mathrm{diag}} := (P_{\nu,jj})_{j=1}^N$. The \emph{two-body stability operator} 
is given by $\mathcal{B}_{z,\tilde{z}} := \mathbf{1}-m_\smc(z)m_\smc(\tilde{z}) S$
and its inverse can be explicitly computed. 

The proof of the following Lemma \ref{lem: inv stab op 3} is a direct computation and hence omitted. 
\begin{lemma}[Inverse of the stability operator]\label{lem: inv stab op 3}
	Let $z,\tilde{z}\in\mathbb{C}$ and denote  $\mathfrak{m} = m_\smc(z) m_\smc(\tilde{z})$. Let $\alpha:=1+\lambda$, $\beta:=1-\lambda^3$ and $\gamma:=1+\lambda+\lambda^2$. Then the inverse of the two-body stability operator is given by
\footnotesize
\begin{align}
	\mathcal{B}_{z,\tilde{z}}^{-1} &=  \mathbf{1} - \frac{1}{D\hat{n} n_2 n_3 (1-\mathfrak{m})(1-\mathfrak{m}/n_2)(1-\mathfrak{m}/n_3)}  
\left(\begin{matrix}
	\hat{b}_{11} E_{d_1,d_1}& \hat{b}_{12} E_{d_1,d_2} & \hat{b}_{13} E_{d_1,d_3}\\
	\hat{b}_{21} E_{d_2,d_1} & \hat{b}_{22} E_{d_2,d_2} & \hat{b}_{23} E_{d_2,d_3}\\
	\hat{b}_{31} E_{d_3,d_1} & \hat{b}_{32} E_{d_3,d_2} & \hat{b}_{33}E_{d_3,d_3}\\
\end{matrix}\right) \, ,\\
&=:\mathbf{1} + \left(\begin{matrix}
	b_{11} E_{d_1,d_1}& b_{12} E_{d_1,d_2} & b_{13} E_{d_1,d_3}\\
	b_{21} E_{d_2,d_1} & b_{22} E_{d_2,d_2} & b_{23} E_{d_2,d_3}\\
	b_{31} E_{d_3,d_1} & b_{32} E_{d_3,d_2} & b_{33}E_{d_3,d_3}\\
\end{matrix}\right)\nonumber
\end{align}
\normalsize
where
\begin{align}
	\label{eq:nj} n_{2,3} = \frac{2-\lambda^2-\lambda^3\mp\sqrt{(1+\lambda)(4\lambda^3 + \lambda^4 + \lambda^5)}}{2(1-\lambda-\lambda^2)}, \qquad \hat{n} = \alpha\lambda^3(\gamma+\lambda)+\lambda^3 - \alpha\beta.
\end{align}
Moreover, the matrix entries $\hat{b}_{\mu\nu}$ are given by
\footnotesize
\begin{align*}
	\hat{b}_{11} &= \frac{1}{\gamma-\mathfrak{m}} \left[\frac{\mathfrak{m}}{\lambda^2} + \mathfrak{m}^2\alpha^2\lambda\gamma(\gamma-\mathfrak{m}\alpha)\right],\qquad
	\hat{b}_{12}=\hat{b}_{21} = \mathfrak{m}\alpha\gamma(\gamma-\mathfrak{m}\alpha),\\
	\hat{b}_{22} &= \frac{1}{R} \left[\left(\frac{\mathfrak{m}\beta}{\lambda}(\gamma-\mathfrak{m}) + (\mathfrak{m}\alpha\lambda)^2\right)\left(R(\gamma-\mathfrak{m}\alpha) - \mathfrak{m}^2\lambda^3(\gamma-\mathfrak{m})\right) + \mathfrak{m}^2\lambda^2\gamma(\gamma-\mathfrak{m})^2\right],\\
	\hat{b}_{13} &= \hat{b}_{31} = \mathfrak{m}^2 \alpha\lambda^2\gamma,\qquad
	\hat{b}_{23}=\hat{b}_{32} = \mathfrak{m}\lambda\gamma(\gamma-\mathfrak{m}),\qquad
	\hat{b}_{33} = \mathfrak{m}\alpha\left[(\gamma-\mathfrak{m})(\gamma-\mathfrak{m}\beta) - (\mathfrak{m}\alpha)^2\lambda^3\right] + \mathfrak{m}^2\lambda^3(\gamma-\mathfrak{m}),
\end{align*}
\normalsize
where we abbreviated
\begin{align*}
	R=(\gamma-\mathfrak{m}\beta)(\gamma-\mathfrak{m}) -(\mathfrak{m}\alpha)^2\lambda^3.
\end{align*}
	
\end{lemma}
We remark that $n_{2,3}>1$ for $\lambda>0$ small enough and it holds that $n_{2,3} \downarrow 1$ as $\lambda\downarrow 0$.

In this way, we find that
\begin{align}
	\langle P_\mu M(z,P_\nu,\tilde{z})\rangle  =\mathfrak{m}\left(\delta_{\mu \nu} \frac{d_{\mu}}{N} + \frac{d_\mu d_\nu}{N} b_{\mu \nu}\right)
\end{align}
with $b_{\mu \nu} = b_{\mu \nu}(z, \tilde{z})$ and $\mathfrak{m} = \mathfrak{m}(z, \tilde{z})$ given in Lemma \ref{lem: inv stab op 3}, and therefore by \eqref{eq:globallaw3}
\begin{align*}
	\langle P_\mu G(z) P_\nu G(\tilde{z})\rangle = \mathfrak{m}\left(\delta_{\mu \nu} \frac{d_{\mu}}{N} + \frac{d_\mu d_\nu}{N} b_{\mu \nu}\right) + \mathcal{O}_\prec(C(\lambda, t)/N).
\end{align*}
Plugging this into \eqref{eq: wmunu PGPG} and using that $|\ee^{\ii t(z-\tilde{z})}| \lesssim 1$ for $z \in \gamma$ and $\tilde{z} \in \tilde{\gamma}$, leads to
\begin{equation} \label{eq:wmunugeneral}
	w_{\mu \nu}(t) = \frac{1}{(2\pi \ii)^2} \oint_{\gamma}\oint_{\tilde{\gamma}} \ee^{\ii t(z-\tilde{z})} \mathfrak{m}(z, \tilde{z})\left(\delta_{\mu \nu} + d_\nu b_{\mu \nu}(z, \tilde{z})\right)\, \dd z\, \dd \tilde{z} + \mathcal{O}_\prec(C(\lambda, t)/N).
\end{equation}

To evaluate this, we make use of the following lemma, whose proof is given at the end of this section.

\begin{lemma}[Contour integrals of $m_\smc$]\label{lem: integral fraction 3}
	Let $r_1,r_2 \in (0,1)$ and let $\mathfrak{m} = \mathfrak{m}(z, \tilde{z})$ be as in Lemma~\ref{lem: inv stab op 3}. Then, for $j =2,3$, we have that 
	\begin{equation} \label{eq: Aj}
		\begin{split}
		A_j(r_1,r_2) :=&\frac{1}{(2\pi \ii)^2} \oint_{\gamma}\oint_{\tilde{\gamma}} \ee^{\ii t(z-\tilde{z})} \frac{\mathfrak{m}^j}{(1-\mathfrak{m})(1-r_1\mathfrak{m})(1-r_2\mathfrak{m})}\, \dd z\, \dd \tilde{z}\\ =& \frac{1}{(1-r_1)(1-r_2)} - \frac{1}{\pi t^3(r_1 - r_2)} \left[ \frac{1 + o(1)}{(1 - r_1)^4} - \frac{1 + o(1)}{(1 - r_2)^4} \right] ,\\
		\widehat{A}(r_1,r_2) :=&\frac{1}{(2\pi \ii)^2} \oint_{\gamma}\oint_{\tilde{\gamma}} \ee^{\ii t(z-\tilde{z})} \frac{\mathfrak{m}^2}{(1-r_1\mathfrak{m})(1-r_2\mathfrak{m})}\, \dd z\, \dd\tilde{z}\\
		=& \frac{1}{\pi t^3 (r_1-r_2)} \left[\frac{1+o(1)}{(1-r_1)^3}-\frac{1+o(1)}{(1-r_2)^3}\right],
	\end{split}
	\end{equation}
	where $o(1)$ vanishes as $t \to \infty$,  $r_1, r_2 \uparrow 1$, and $t(1 - r_1) \to \infty$ and $t (1 - r_2) \to \infty$. 
\end{lemma}

Armed with \eqref{eq:wmunugeneral} and Lemma \ref{lem: integral fraction 3}, we now turn to the computation of the $w_{\mu \nu}$'s from \eqref{eq:choiceofws}. 
\\[2mm]
\underline{Computation of $w_{12}(t)$:}
From \eqref{eq:wmunugeneral} and Lemma~\ref{lem: inv stab op 3} we have
\footnotesize
\begin{align}\label{eq:w12}
	w_{12}(t)&= \frac{-\gamma\alpha\lambda}{\hat{n}n_2n_3} \frac{1}{(2\pi\ii)^2} \oint_{\gamma}\oint_{\tilde{\gamma}} \ee^{\ii t(z-\tilde{z})} \frac{-\alpha\mathfrak{m}^3 + \gamma\mathfrak{m}^2}{ (1-\mathfrak{m})(1-\mathfrak{m}/n_2)(1-\mathfrak{m}/n_3)}\,\dd z\,\dd \tilde{z}+ \mathcal{O}_\prec(C(\lambda,t)/N),\nonumber
\end{align}
\normalsize
with $\mathfrak{m} = m_\smc(z) m_\smc(\tilde{z})$ and $\alpha,\gamma,n_2,n_3$ and $\hat{n}$ defined in Lemma~\ref{lem: inv stab op 3}. The integrand in \eqref{eq:w12} can be written as
\begin{align*}
	\frac{-\alpha\mathfrak{m}^3 + \gamma\mathfrak{m}^2}{ (1-\mathfrak{m})(1-\mathfrak{m}/n_2)(1-\mathfrak{m}/n_3)} = \frac{(1+\lambda)\mathfrak{m}^2}{(1-\mathfrak{m}/n_2)(1-\mathfrak{m}/n_3)} + \frac{\lambda^2\mathfrak{m}^2}{(1-\mathfrak{m})(1-\mathfrak{m}/n_2)(1-\mathfrak{m}/n_3)}
\end{align*}
and hence we obtain
\begin{align}
	w_{12}(t) = -\gamma\alpha\lambda \frac{(1+\lambda)\widehat{A}(n_2^{-1},n_3^{-1}) + \lambda^2 A_2(n_2^{-1},n_3^{-1})}{\hat{n}n_2 n_3} + \mathcal{O}_\prec(C(\lambda,t)/N),
\end{align}
where $\widehat{A}$ and $A_2$ are defined in Lemma~\ref{lem: inv stab op 3}. For $n_{2,3}$ as in \eqref{eq:nj} we find for $j=2,3$ that
\footnotesize
\begin{align*}
	n_2^{-1} n_3^{-1} A_j(n_2^{-1},n_3^{-1}) = \frac{1}{(n_2-1)(n_3-1)} - \frac{4}{\pi t^3 \lambda^5} (1+o(1)), \quad n_2^{-1} n_3^{-1} \widehat{A}(n_2^{-1},n_3^{-1}) = \frac{3}{\pi t^3\lambda^4}(1+o(1)),
\end{align*}
\normalsize
where $o(1)$ vanishes as $t\to\infty$, $\lambda\to 0$, and $t\lambda\to\infty$. 
Plugging this back in \eqref{eq:w12}, we conclude
\begin{align}
	w_{12}(t) = \frac{d_2}{N} + \frac{3}{\pi(\lambda t)^3}(1+o(1)) + \mathcal{O}_\prec(C(\lambda,t)/N). \label{eq:w12final}
\end{align}
\underline{Computation of $w_{13}(t)$:} With \eqref{eq:wmunugeneral} and Lemma~\ref{lem: inv stab op 3} we find that 
\footnotesize
\begin{align}
	w_{13}(t) = \frac{-\alpha\gamma\lambda^2}{\hat{n}n_2 n_3}\frac{1}{(2\pi\ii)^2} \oint_{\gamma}\oint_{\tilde{\gamma}}\ee^{\ii t(z-\tilde{z})} \frac{\mathfrak{m}^3}{(1-\mathfrak{m})(1-\mathfrak{m}/n_2)(1-\mathfrak{m}/n_3)}\, \dd z\,\dd \tilde{z} + \mathcal{O}_\prec(C(\lambda,t)/N).
\end{align}
\normalsize
With the help of Lemma~\ref{lem: integral fraction 3} we obtain
\begin{align}
	w_{13}(t) = -\alpha\gamma\lambda^2 \frac{A_3(n_2^{-1},n_3^{-1})}{\hat{n}n_2 n_3} + \mathcal{O}_\prec(C(\lambda,t)/N)
\end{align}
and from this we conclude, just as in the computation of $w_{12}(t)$, that
\begin{align}
	w_{13}(t) = \frac{d_3}{N} - \frac{4}{\pi(\lambda t)^3}(1+o(1))+\mathcal{O}_\prec(C(\lambda,t)/N).\label{eq:w13final}
\end{align}
\underline{Computation of $w_{23}(t)$:} Again, it follows from \eqref{eq:wmunugeneral} together with Lemma~\ref{lem: inv stab op 3} that
\footnotesize
\begin{align}
	w_{23}(t) = - \frac{\lambda\gamma}{\hat{n}n_2 n_3} \frac{1}{(2\pi\ii)^2}\oint_{\gamma}\oint_{\tilde{\gamma}} \ee^{\ii t(z-\tilde{z}) } \frac{\gamma\mathfrak{m}^2 - \mathfrak{m}^3}{(1-\mathfrak{m})(1-\mathfrak{m}/n_2)(1-\mathfrak{m}/n_3)}\, \dd z\, \dd\tilde{z} + \mathcal{O}_\prec(C(\lambda,t)/N).
\end{align}
\normalsize
The integrand can be expressed as
\begin{align*}
	 \frac{\gamma\mathfrak{m}^2 - \mathfrak{m}^3}{(1-\mathfrak{m})(1-\mathfrak{m}/n_2)(1-\mathfrak{m}/n_3)} = \frac{\mathfrak{m}^2}{(1-\mathfrak{m}/n_2)(1-\mathfrak{m}/n_3)} + \frac{\lambda(1+\lambda)\mathfrak{m}^2}{(1-\mathfrak{m})(1-\mathfrak{m}/n_2)(1-\mathfrak{m}/n_3)}
\end{align*}
and by using Lemma~\ref{lem: integral fraction 3} we find that
\begin{align}
	w_{23}(t) = - \lambda\gamma \frac{\widehat{A}(n_2^{-1},n_3^{-1})+\lambda(1+\lambda)A_2(n_2^{-1},n_3^{-1})}{\hat{n}n_2 n_3}+\mathcal{O}_\prec(C(\lambda,t)/N).
\end{align}
As before, we finally arrive at
\begin{align}
	w_{23}(t) = \frac{d_3}{N}-\frac{1}{\pi(\lambda t)^3}(1+o(1))+\mathcal{O}_\prec(C(\lambda,t)/N).\label{eq:w23final}
\end{align}
Combining the results in \eqref{eq:w12final}, \eqref{eq:w13final}, and \eqref{eq:w23final} with the relations in Lemma~\ref{lem: rel wmunu3} together with the symmetry $t\mapsto -t$ (since $-H$ is again a matrix satisfying Assumption~\ref{ass:3}), we obtain all nine $w_{\mu\nu}$'s for $\mu,\nu\in\{1,2,3\}$. In Theorem~\ref{thm: wmunu3} we only record the off-diagonal ones. This concludes the proof. \qed

We are left with giving the proof of Lemma \ref{lem: integral fraction 3}. 

\begin{proof}[Proof of Lemma \ref{lem: integral fraction 3}]
We provide the details for the computation of $A_2$, the cases of $A_3$ and $\widehat{A}$ are analogous and hence kept brief. First, since $|\mathfrak{m}| < 1$ for $z \in \gamma$ and $\tilde{z} \in \tilde{\gamma}$, we can write 
\begin{align*}
	&\frac{\mathfrak{m}^2}{(1-\mathfrak{m})(1-r_1\mathfrak{m})(1-r_2\mathfrak{m})} = \sum_{k_1,k_2,k_3 \geq 0} \mathfrak{m}^{k_1+2} (r_1\mathfrak{m})^{k_2} (r_2\mathfrak{m})^{k_3}\\
	=& \frac{1}{(1-r_1)(1-r_2)}\sum_{l\geq 0} \mathfrak{m}^{l+2} - \frac{1}{(1-r_1)(r_1 - r_2)}\sum_{l\geq 0} (r_1\mathfrak{m})^{l+2} + \frac{1}{(1-r_2)(r_1 - r_2)}\sum_{l\geq 0} (r_2\mathfrak{m})^{l+2}.
\end{align*}
and similarly
\begin{align*}
	\frac{\mathfrak{m}^2}{(1-r_1\mathfrak{m})(1-r_2\mathfrak{m})} 
	&= \frac{1}{r_1(r_1-r_2)} \sum_{l\geq 0} (r_1\mathfrak{m})^{l+2} - \frac{1}{r_2(r_1-r_2)} \sum_{l\geq 0} (r_2\mathfrak{m})^{l+2}.
\end{align*}
Hence, there are three sums to evaluate. 

For the first sum, with the aid of Lemmas~\ref{lem: contour int}--\ref{lem: sum prod J_k}, we find that 
\begin{equation} \label{eq:const}
	\begin{split}
		&\frac{1}{(2\pi \ii)^2} \oint_{\gamma}\oint_{\tilde{\gamma}} \ee^{\ii t(z-\tilde{z})} \sum_{l\geq 0} \mathfrak{m}^{l+2}\, \dd z\, \dd\tilde{z} \\
		= &\sum_{l\geq 0} \left[\ii^{l+3}J_{l+3}(-2t)-\ii^{l+1} J_{l+1}(-2t)\right] \left[\ii^{l+3}J_{l+3}(2t) - \ii^{l+1} J_{l+1}(2t)\right] = 1 - \frac{J_1(2t)^2}{t^2} \,. 
	\end{split}
\end{equation}
From \eqref{eq:T2} we find that
\begin{align}
	\frac{1}{(2\pi \ii)^2}\oint_{\gamma}\oint_{\tilde{\gamma}} \ee^{\ii t(z-\tilde{z})}&\sum_{l\geq 0}\left(r_j\mathfrak{m}\right)^{l+2}\, \dd z\, \dd \tilde{z}= \frac{1}{\pi} \frac{1}{((1-r_j)t)^3}(1+o(1)),\label{eq:timedep}
\end{align}
where $o(1)$ vanishes as $t\to\infty$, $r_j\uparrow 1$ and $t(1-r_j)\to \infty$.
The results for $A_2$ and $\widehat{A}$ immediately follow from these computations.

To obtain the result for $A_3$, we compute, for $j=1,2$, 
\begin{alignat*}{2}
	&\frac{1}{(2\pi i)^2}\oint_{\gamma}\oint_{\tilde{\gamma}} \ee^{\ii t(z-\tilde{z})}\sum_{l\geq 0} \mathfrak{m}^{l+3}\, \dd z\, \dd\tilde{z} &&=  1- \frac{J_1(2t)^2}{t^2}-4\frac{J_2(2t)^2}{t^2},\\
	&\frac{1}{(2\pi \ii)^2} \oint_{\gamma}\oint_{\tilde{\gamma}} \ee^{\ii t(z-\tilde{z})} \sum_{l\geq 0}\left(r_j\mathfrak{m}\right)^{l+3}\,\dd z\, \dd\tilde{z} &&= \frac{1}{\pi} \frac{1}{((1- r_j)t)^3} (1 + o(1)) ,
\end{alignat*}
where again $o(1)$ vanishes as $t \to \infty$, $r_j \uparrow 1$ and $t (1-r_j) \to \infty$. 
From this the result for $A_3$ follows immediately.
\end{proof} 
\subsection{Proof of Theorem~\ref{thm: Pnupsi3}} With the same modifications as needed in the proof of Theorem \ref{thm: Pnupsi} in Section \ref{sec:2proof}, the argument works exactly the same as in the proof of Theorem \ref{thm: dappmunu3}. We leave the details to the reader. \qed

\appendix

\section{Additional proofs} \label{app:aux}
\subsection{Proof of Lemma \ref{lem: contour int}}
This is a direct computation. Since $m_\smc(z)$ is holomorphic in $\C \setminus [-2,2]$, we can deform the contour $\gamma$ to be $ \gamma = [-2,2] \pm \ii 0$. With the help of the parametrization $m_\smc(z)=\exp(\pm \ii\theta)$ for $z \in \gamma$ where $\theta = \theta(z)$ is defined by $z=-2\cos(\theta)$, we find that 
	\begin{align*}
		\oint_{\gamma} \ee^{\ii tz} m_\smc^{n}(z)\, \dd z &= \int_{-2}^2 \ee^{\ii tz} \ee^{-\ii\theta(z)n}\, \dd z - \int_{-2}^2 e^{itz} e^{i\theta(z)n}\, \dd z\\
		&= -2\ii \int_{-2}^2 \ee^{\ii tz} \sin(\theta(z)n)\, \dd z\\
		&= -4\ii \int_{0}^\pi \ee^{-2\ii t\cos(\theta)} \sin(\theta n) \sin(\theta)\, \dd \theta\\
		&= 2\ii \int_0^\pi \ee^{–2\ii t\cos(\theta)} \left(\cos(\theta(n+1))-\cos(\theta(n-1))\right) \dd\theta\\
		&= 2\pi \ii \left(\ii^{n+1} J_{n+1}(-2t)- \ii^{n-1}J_{n-1}(-2t)\right) \,. 
	\end{align*}
In the last step we used that \cite[9.1.21]{AS72}
\begin{align*}
	J_n(z) = \frac{\ii^{-n}}{\pi} \int_0^\pi \ee^{\ii z\cos(\theta)} \cos(n\theta)\, \dd \theta \,. 
\end{align*}
Dividing by $2\pi \ii$ gives the result. \qed

\subsection{Proof of Lemma~\ref{lem: sum prod J_k}}
	The sums can be computed with the help of Graf's addition theorem \cite[Eq. 9.1.79]{AS72} and some symmetries of the Bessel functions. Graf's addition theorem states that for any non-zero constant $c\in\mathbb{C}$, $x,y\in\mathbb{C}$ and $v\in\mathbb{N}_0$,
	\begin{align}
		\sum_{n=-\infty}^{\infty} c^n I_{n+v}(x) I_n(y) = \left(\frac{x+yc^{-1}}{x+yc}\right)^{v/2} I_v \left(\sqrt{x^2+y^2+xy(c+c^{-1})}\right), \label{eq: Graf}
	\end{align}
	where $I_n$ denotes the $n$-th modified Bessel function. The relation between $I_n$ and $J_n$ is given by $I_n(x)=\ii^{-n}J_n(\ii x)$.
	
	If $c=1$, we can use symmetries to turn \eqref{eq: Graf} into a formula for a sum which is only over the non-negative integers. Suppose that $v=0$ in \eqref{eq: Graf}. Then using  $J_{-n}(x)=(-1)^n J_n(x)$ and the relation between $I_n$ and $J_n$ we have that $I_{-n}(x) = (-1)^n \ii^{2n} I_n(x)$ and hence obtain
	\begin{align}
		\sum_{n\geq 0} I_n(x) I_n(y) = \frac{1}{2}\left(I_0(x)I_0(y)+I_0(|x+y|)\right). \label{eq: Bessel 1}
	\end{align}
	This immediately implies \eqref{eq:sumprod1}. 
	
	Next, a simple computation shows that 
		\begin{equation*}
			\sum_{n\geq 0} I_{n+p}(x) I_{n+p}(-x)= \frac{1}{2}\left(1-J_0(ix)^2\right)-\sum_{k=1}^{p-1} J_k(ix)^2\,, 
		\end{equation*}
	where we used \eqref{eq: Bessel 1} and that $I_0(0)=1$ and $J_k(-x) = (-1)^k J_k(x)$.
	Substituting $x\to x/\ii$ and using again the relation between $I_n$ and $J_n$ yields \eqref{eq:sumprod2}. 
	
To prove \eqref{eq:subprod3}, we note that 
		\begin{equation*}
			\sum_{n=-\infty}^{\infty} I_{n+q}(x) I_n(-x) = 2 \sum_{n\geq 0} I_{n+q}(x) I_n(-x) + \sum_{k=0}^{q-1} I_k(x) I_{k-q}(-x) - I_q(x) I_0(x)
		\end{equation*}
which yields
		\begin{equation*}
			\sum_{n\geq 0} I_{n+q}(x) I_n(-x)=\frac{1}{2}\left(I_q(x) I_0(x) - \sum_{k=0}^{q-1}I_k(x) I_{k-q}(-x)\right),
		\end{equation*}
by application of Graf's addition theorem \eqref{eq: Graf}. After rewriting the $I_n$ in terms of $J_n$ and replacing $x$ by $x/\ii$ we conclude \eqref{eq:subprod3}. \qed

\subsection*{Acknowledgments.}
L.E.~and J.H.~are supported by the ERC Advanced Grant ``RMTBeyond'' No.~101020331. Moreover, J.H.~acknowledges (partial) financial support by the ERC Consolidator Grant ``ProbQuant'' (jointly with the Swiss State Secretariat for Education, Research and Innovation). C.V.~was (partially) supported by the German Academic Scholarship Foundation and the Deutsche Forschungsgemeinschaft (DFG, German Research Foundation) -- TRR 352 -- Project-ID 470903074. Moreover, C.V.~acknowledges (partial) financial support by the ERC Starting Grant ``FermiMath" No.~101040991 and the ERC Consolidator Grant ``RAMBAS'' No. 10104424, funded by the European Union. Views and opinions expressed are however those of the authors only and do not necessarily reflect those of the European Union or the European Research Council Executive Agency. Neither the European Union nor the granting authority can be held responsible for them.

\subsection*{Data Availability Statement.} The Matlab code used to generate the datasets of the provided examples is available from the corresponding author on request. 

\subsection*{Conflict of Interest Statement.} The authors have no conflicts of interest.

\bibliographystyle{plainurl}
\bibliography{Literature.bib}

@article{vonNeumann29,
    author = {J. von Neumann},
    title = {{Beweis des Ergodensatzes und des $H$-Theorems in der neuen Mechanik}},
    journal = {Zeitschrift f\"ur Physik},
    volume = {57},
    pages = {30--70},
    year = {1929},
    note = {English translation: \textit{European Physical Journal H}, 35: 201--237, 2010},
}

@article{GLMTZ10,
	author = {S. Goldstein and J. L. Lebowitz and C. Mastrodonato and R. Tumulka and N. Zangh\`\i},
	title = {{Normal Typicality and von Neumann's Quantum Ergodic Theorem}},
	journal = {Proceedings of the Royal Society A},
	pages = {3203--3224},
	volume = {466},
	number = {2123},
	year = {2010}
}

@article{GLMTZ10b,
	author = {S. Goldstein and J. L. Lebowitz and C. Mastrodonato and R. Tumulka and N. Zangh\`\i},
	title = {{On the Approach to Thermal Equilibrium of Macroscopic Quantum Systems}},
	journal = {Physical Review E},
	volume = {81},
	pages = {011109},
	year = {2010}
}

@article{GLTZ10,
	author = {S. Goldstein and J. L. Lebowitz and R. Tumulka and N. Zangh\`\i},
	title = {{Long-Time Behavior of Macroscopic Quantum Systems}},
	journal = {European Physical Journal H},
	volume = {35},
	pages = {173--200},
	year = {2010}
}

@article{GHT13,
	author = {S. Goldstein and T. Hara and H. Tasaki},
	title = {{Time Scales in the Approach to Equilibrium of Macroscopic Quantum Systems}},
	journal = {Physical Review Letters},
	volume = {111},
	pages = {140401},
	year = {2013}
}

@article{GHT14,
	author = {S. Goldstein and T. Hara and H. Tasaki},
	title = {The approach to equilibrium in a macroscopic quantum system for a typical nonequilibrium subspace},
	year = {2014},
	journal = {arXiv:1402.3380}
}

@article{GHT15,
	author = {S. Goldstein and T. Hara and H. Tasaki},
	title = {{Extremely quick thermalization in a macroscopic quantum system for a typical nonequilibrium subspace}},
	journal = {New Journal of Physics},
	volume = {17},
	pages = {045002},
	year = {2015}
}

@article{TTV23a,
	author = {S. Teufel and R. Tumulka and C. Vogel},
	title = {{Time Evolution of Typical Pure States from a Macroscopic Hilbert Subspace}},
	journal = {Journal of Statistical Physics},
	volume = {190},
	pages = {69},
	year = {2023}
}

@article{TTV23b,
	author = {S. Teufel and R. Tumulka and C. Vogel},
title = {{Typical Macroscopic Long-Time Behavior for Random Hamiltonians}},
journal = {Annales Henri Poincaré},
volume = {26},
pages = {3189--3231},
year = {2025}
}

@book{AS72,
	author = {M. Abramowitz and I. A. Stegun},
	title = {Handbook of Mathematical Functions With Formulas, Graphs, and Mathematical Tables},
	year = {1972}
}

@incollection{BRGSR18,
	author = {Balz, B. and Richter, J. and Gemmer, J. and Steinigeweg, R. and Reimann, P.},
	title = {{Dynamical typicality for initial states with a preset measurement statistics of several commuting observables}},
	booktitle = {{Thermodynamics in the Quantum Regime}},
	editor = {F. Binger and L. A. Correa and C. Gogolin and J. Anders and G. Adesso},
	publisher = {Springer, Cham},
	year = {2019},
	chapter = {17},
	pages = {413--433}
}

@article{MGE11,
	author = {M. P. M\"uller and D. Gross and J. Eisert},
	title = {Concentration of Measure for Quantum States with a Fixed Expectation Value},
	journal = {Communications in Mathematical Physics},
	volume = {303},
	pages = {785--824},
	year = {2011}
}

@article{BG09,
	author = {C. Bartsch and J. Gemmer},
	title = {Dynamical Typicality of Quantum Expectation Values},
	journal = {Physical Review Letters},
	volume = {102},
	pages = {110403},
	year = {2009}
}

@article{RG20,
	author = {P. Reimann and J. Gemmer},
	title = {Why are macroscopic experiments reproducible? {I}mitating the behavior of an ensemble by single pure states},
	journal = {Physica A},
	volume = {552},
	pages = {121840},
	year = {2020}
}

@article{Reimann18a,
	author = {P. Reimann},
	title = {Dynamical typicality of isolated many-body quantum systems},
	journal = {Physical Review E},
	volume = {97},
	pages = {062129},
	year = {2018}
}

@article{Reimann18b,
	author = {P. Reimann},
	title = {Dynamical Typicality Approach to Eigenstate Thermalization},
	journal = {Physical Review Letters},
	volume = {120},
	pages = {230601},
	year = {2018}
}

@article{ER24,
	author = {L. Erdős and V. Riabov},
	title = {{Eigenstate Thermalization Hypothesis for Wigner-type Matrices}},
	year = {2024},
  journal={Communications in Mathematical Physics},
  volume={405},
  number={12},
  pages={282},
  publisher={Springer}
}

@misc{AEK19,
	author = {O. Ajanki and L. Erdős and T. Krüger},
	title = {{Quadratic Vector Equations On Complex Upper Half-Plane}},
	howpublished ={In: \textsl{Mem. Am. Math. Soc.} 261},
	year = {2019}
}

@article{AEK17,
	author = {O. Ajanki and L. Erdős and T. Krüger},
	title = {{Universality for general Wigner-type matrices}},
	journal = {Probability Theory and Related Fields},
	volume = {169},
	pages = {667--727},
	year = {2017}
}

@article{SF12,
	author = {A. J. Short and T. C. Farrelly},
	title = {Quantum equilibration in finite time},
	journal = {New Journal of Physics},
	volume = {14},
	pages = {013063},
	year = {2012}
}

@article{edgeETH,
	title        = {{Eigenstate thermalisation at the edge for Wigner matrices}},
	author       = {Cipolloni, G. and Erd\H{o}s, L. and Henheik, J.},
	year         = 2023,
	journal      = {arXiv:2309.05488}
}

@article{OTOC,
	title        = {{Out-of-time-ordered correlators for Wigner matrices}},
	author       = {Cipolloni, G. and Erd\H{o}s, L. and Henheik, J.},
	year         = 2024,
	journal      = {Adv. Theor. Math. Phys.},
	volume = {28},
	number = {6},
	pages = {2025--2083}
}

@article{A2,
	title        = {Rank-uniform local law for {Wigner} matrices},
	author       = {Cipolloni, G. and Erd\H{o}s, L. and Schröder, D.},
	year         = 2022,
	journal      = {Forum of Mathematics, Sigma},
	volume       = 10,
	pages        = {E96}
}

@article{multiG,
	title        = {Optimal multi-resolvent local laws for {Wigner} matrices},
	author       = {Cipolloni, G. and Erd\H{o}s, L. and Schröder, D.},
	year         = 2022,
	journal      = {Electron. J. Probab.},
	volume       = 27,
	pages        = {1--38}
}

@article{ETHpaper,
	title        = {{Eigenstate Thermalisation Hypothesis for Wigner matrices}},
	author       = {Cipolloni, G. and Erdős, L. and Schröder, D.},
	year         = 2021,
	journal      = {Comm. Math. Phys.},
	volume       = 388,
	pages        = {1005--1048}
}

@article{decor,
	title        = {Eigenvector decorrelation for random matrices},
	author       = {Cipolloni, G. and Erd\H{o}s, L. and Henheik, J. and Kolupaiev, O.},
	year         = 2024,
	journal         = {arXiv:2410.10718}
}

@article{equipart,
	title        = {{Gaussian fluctuations in the equipartition principle for {Wigner} matrices}},
	author       = {Cipolloni, G. and Erd\H{o}s, L. and Henheik, J. and Kolupaiev, O.},
	year         = 2023,
	journal      = {Forum Math., Sigma},
	volume       = 11,
	pages       = {E74}
}

@article{deutsch,
	title        = {Quantum statistical mechanics in a closed system},
	author       = {Deutsch, J.~M.},
	year         = 1991,
	journal      = {Phys. Rev. A},
	volume       = 43,
	pages        = {2046--2049}
}

@article{Srednicki,
	title        = {Chaos and quantum thermalization},
	author       = {Srednicki, M.},
	year         = 1994,
	journal      = {Phys. Rev. E},
	volume       = 50,
	pages        = {888--901}
}

@article{echo,
	title        = {{Loschmidt echo for deformed Wigner matrices}},
	author       = {Erd\H{o}s, L. and Henheik, J. and Kolupaiev, O.},
  journal={{Lett. Math. Phys.}},
volume={115},
number={1},
pages={1--42},
year={2025},
publisher={Springer}
}

@article{pretherm,
	author = {Erd\H{o}s, L. and Henheik, J. and Reker, J. and Riabov, V.},
	title = {{Prethermalization for Deformed Wigner Matrices}},
  journal ={Ann. Henri Poincar{\'e}},
  volume = {26},
pages={1991--2033},
year={2025}
}

@book{Gradshteyn.Ryzhik.2007,
	title        = {{Table of Integrals, Series, and Products}},
	author       = {Gradshteyn, I. S. and Ryzhik, I. M.},
	year         = 2007,
	publisher    = {Academic Press},
	address      = {Boston},
	doi          = {https://doi.org/10.1016/C2009-0-22516-5},
	isbn         = {978-0-12-373637-6},
	edition      = {{Seventh}}
}

@article{semicirclegeneral,
	title        = {The local semicircle law for a general class of random matrices},
	author       = {Erd\H{o}s, L. and Knowles, A. and Yau, H.-T. and Yin, J.},
	year         = 2013,
	journal      = {Electron. J. Probab.},
	volume       = 18,
	number       = 59,
	pages        = {1--58}
}
\end{document}